\documentclass[11pt, letterpaper]{article}
\usepackage{fullpage}
\usepackage[utf8]{inputenc} \usepackage{ae}
\usepackage{enumerate}
\usepackage{amssymb}
\usepackage{amsmath}
\usepackage{graphicx}
\usepackage{stmaryrd}
\usepackage[breaklinks,bookmarks=false]{hyperref}
\usepackage{enumerate}
\usepackage[usenames,dvipsnames,svgnames,table]{xcolor}

\newtheorem{theorem}{Theorem}[section]

\newtheorem{lemma}[theorem]{Lemma}

\newtheorem{definition}[theorem]{Definition}

\newtheorem{fact}[theorem]{Fact}

\newenvironment{proof}{\noindent{\bf Proof:}\hspace*{1em}}{\qed\bigskip}
 
\newcommand{\qed}{\hfill\ensuremath{\square}}

\DeclareMathOperator{\age}{age}
\DeclareMathOperator{\dist}{dist}

\newcommand{\cC}{\mathcal{C}}
\newcommand{\ocC}{\bar{\mathcal{C}}}
\newcommand{\cP}{\mathcal{P}}
\newcommand{\tO}{\widetilde{O}}
\newcommand{\Dstar}{\gamma^*}
\newcommand{\cT}{\mathcal{T}}

\newcommand{\Reff}{R_{\mathrm{eff}}}
\newcommand{\diameff}{\gamma_{\mathrm{eff}}}
\newcommand{\diamdist}{\Delta}
\newcommand{\energy}{\mathcal{E}}
\newcommand{\cut}{{K}}
\newcommand{\wit}[1]{\mathop{W}(#1)}
\newcommand{\cR}{\mathcal{R}}
\newcommand{\of}{\bar{f}}

\newcommand{\bE}{\mathbb{E}}

\begin{document}

\title{Fast Generation of Random Spanning Trees and the Effective Resistance Metric}

\author{ Aleksander Mądry\thanks{aleksander.madry@epfl.ch} \\ EPFL \and Damian Straszak\thanks{damian.straszak@gmail.com} \\ University of Wrocław\and Jakub Tarnawski\thanks{jakub.tarnawski@gmail.com} \\ University of Wrocław}

\date{}
 
\maketitle

\begin{abstract}
We present a new algorithm for generating a uniformly random spanning tree in an undirected graph. Our algorithm samples such a tree in expected $\tO(m^{4/3})$ time. This improves over the best previously known bound of $\min(\tO(m\sqrt{n}),O(n^{\omega}))$ -- that follows from the work of Kelner and Mądry [FOCS'09] and of Colbourn et al. [J. Algorithms'96] -- whenever the input graph is sufficiently sparse.

At a high level, our result stems from carefully exploiting the interplay of random spanning trees, random walks, and the notion of effective resistance, as well as from devising a way to algorithmically relate these concepts to the combinatorial structure of the graph. This involves, in particular, establishing a new connection between the effective resistance metric and the cut structure of the underlying graph.  
\end{abstract}

\section{Introduction}\label{sec:introduction}

Random spanning trees are among the oldest and most extensively investigated probabilistic objects in graph theory (see, e.g., \cite{LyonsP13}). Their study goes back as far as to the works of Kirchhoff in the 1840s \cite{Kirchhoff47}. Since then, it has spawned a large number of surprising connections to a variety of topics in theoretical computer science, as well as in discrete probability and other areas of mathematics. For example, random spanning trees are intimately related to the task of spectral graph sparsification \cite{SpielmanS08,BatsonSS09} and Goyal et al. \cite{GoyalRV09} showed how to directly use them to generate an efficient cut sparsifier. They are also closely tied to the notion of electrical flows, a powerful primitive that was instrumental in some of the recent progress on the maximum flow problem \cite{ChristianoKMST11,Madry13}. Finally, random spanning trees were at the heart of the recent results that broke certain long-standing approximation barriers for both the symmetric and the asymmetric version of the Traveling Salesman problem \cite{AsadpourGMOS10,OveisGharanSS11}. 

It is, therefore, not surprising that the task of developing fast algorithms for generating random spanning trees has received considerable attention over the years (\cite{Guenoche83,Kulkarni90,ColbournDM88,ColbournDN89,ColbournMN96,Aldous90,Broder89,Wilson96,KelnerM09}). Broadly, these algorithms can be divided into two categories: determinant-based ones and random walk--based ones. 

The determinant-based algorithms rely on the celebrated Kirchhoff's Matrix Tree Theorem \cite{Kirchhoff47} (also see, e.g., Ch.3, Thm.8 in \cite{Bollobas79}) that reduces counting of the spanning trees of a graph to computing a determinant of a certain associated matrix. This connection was first used in the works of Gu\'enoche \cite{Guenoche83} and of Kulkarni \cite{Kulkarni90} and resulted in $O(mn^3)$-time algorithms. Then, after a sequence of improvements \cite{ColbournDM88,ColbournDN89}, this line of research was culminated with the result of Colbourn et al. \cite{ColbournMN96} that generates a random spanning tree in time $O(n^{\omega})$, where $\omega<2.373$ is the exponent of the fastest algorithm for square matrix multiplication \cite{CoppersmithW90,Vassilevska12}.

On the other hand, the random walk--based approach to random spanning tree generation is based on the following striking theorem due to Broder \cite{Broder89} and Aldous \cite{Aldous90}:

\begin{theorem}[\cite{Aldous90,Broder89}] \label{thm:rand_tree_via_rand_walk}
Suppose you simulate a random walk in an undirected graph $G=(V,E)$, starting from an arbitrary vertex $s$ and continuing until this walk covers the whole graph, i.e., until every vertex has been visited. For each vertex $v$ other than $s$, let $e_v$ be the edge through which $v$ was visited for the first time in this walk. Then, $T = \{ e_v : v \in V \setminus \{s\}\}$ is a uniformly random spanning tree of $G$.
\end{theorem}

In light of this theorem, we can focus our attention on simulating a covering random walk in the input graph $G=(V,E)$ and then just reading off our sampled spanning tree out of the first-visit edges of that walk. The running time of the resulting procedure is proportional to the cover time of the input graph $G$. If $G$ has $n$ vertices and $m$ edges, it is known that this cover time is always $O(mn)$ \cite{AleliunasKLLR79}. This improves upon the $O(n^\omega)$ bound corresponding to the determinant-based approach whenever $G$ is sparse enough. 

Unfortunately, there are examples of graphs for which the above $O(mn)$ bound on the cover time is tight. Therefore, it is natural to wonder if one can obtain a random walk--based algorithm for random spanning tree generation that runs faster than the cover time. Wilson \cite{Wilson96} was the first one to show that this is indeed the case. He put forth a certain random process that is a bit different than random walk, but related. Using this process one is able to generate a random spanning tree in time proportional to the mean hitting time of the graph. This quantity turns out to never be much larger than the cover time and is often much smaller. However, its worst-case asymptotics is still $\Theta(mn)$.

Later on,  Kelner and Mądry \cite{KelnerM09} showed that this worst-case bound of $\Theta(mn)$ can nevertheless be improved upon. (Initially, this improvement was only for approximate sampling, but later Propp \cite{Propp10} showed how to make it work for the exact case too -- see Lemma 7.3.2 in \cite{Madry11}.) They observed that to use Theorem~\ref{thm:rand_tree_via_rand_walk}, one does not need to simulate the covering random walk in full. It suffices to generate its shortcut transcript which retains enough information to allow us to recover what the first-visit edges were. 

 Following this observation, they put forth such a shortcutting technique that relied on a simple diameter--based graph partitioning primitive of Leighton and Rao \cite{LeightonR99}. This technique yields a transcript which is significantly shorter than the full transcript of the covering walk and can be efficiently generated using fast (approximate) Laplacian system solvers \cite{SpielmanT03,SpielmanT04,KoutisMP10,KoutisMP11,KelnerOSZ13}.  This resulted in a generation procedure with $\tO(m\sqrt{n})$ total (expected) running time, which is the best one known for sparse graphs.

\subsection{Our Contribution}

In this paper, we present a new algorithm that makes progress on the sparse-graph regime of fast random spanning tree generation. More precisely, we prove the following theorem. 

\begin{theorem}
\label{thm:main}
For any connected graph $G$ with $m$ edges, we can generate a uniformly random spanning tree of $G$ in expected $\tO(m^{4/3})$ time.
\end{theorem}
This improves over the $\tO(m\sqrt{n})$ bound of Kelner and Mądry \cite{KelnerM09} whenever the graph is sufficiently sparse.

\subsection{Our Approach}\label{sec:our_approach}

The starting point of our approach is the Kelner-Mądry algorithm \cite{KelnerM09} we mentioned above. Recall that this algorithm generates a random spanning tree via simulation of a shortcut transcript of the covering random walk. That transcript retains enough information to allow us to recover what the first-visit edges were while being relatively short and possible to generate quickly. This results in an improved $\tO(m\sqrt{n})$ running time. 

Unfortunately, their technique has certain limitations. Specifically, there are graphs for which the length of the shortcut transcript they aim to construct is inherently $\Omega(n^{3/2})$. (See Section \ref{sec:overview} for such an example.) This makes the $\Omega(n^{3/2})$ running time a barrier for their approach.

Our algorithm broadly follows the theme put forth by Kelner-Mądry. However, to obtain an improved running time and, in particular, to overcome this $\Omega(n^{3/2})$ barrier, we depart from this theme in a number of ways. 

The key new element of our approach is making the effective resistance metric, and its connections to random spanning trees and random walks, the central driver of our algorithm. Compared to the traditional graph distance--based optics employed by Kelner and Mądry, using such an effective resistance--based view gives us a much tighter grasp on the behavior of random walks. It also enables us to tie this behavior to the cut structure of the graph in a new way. 

Specifically, we develop a relationship between the fact that two vertex sets are relatively far away from each other in the effective resistance metric and the existence of a small cut separating them. (See Lemma \ref{lem:good_cut} in Section \ref{sec:parititioning}.)  It can be seen as a certain minimum $s$-$t$ cut--based analogue of the celebrated Cheeger's inequality \cite{AlonM85,Alon86} (see also \cite{KwokLLOT13}) and the role it plays in the context of the sparsest cut problem. We believe this relationship to be of independent interest.

Furthermore, to take full advantage of this effective resistance--based view, we also change the way we approach the task of generation of a random spanning tree. Namely, instead of focusing on efficient generation of a single shortcut transcript and then using it to recover the corresponding random spanning tree in full, we generate a number of carefully crafted transcripts, one after another. None of these transcripts alone is enough to extract the whole random spanning tree from it. More precisely, each one of them provides us only with a sample from the marginal distribution corresponding to the intersection of the random spanning tree with some subset of edges of the graph. However, by making sure that all these samples are consistent with each other (which is achieved by running each consecutive sampling procedure in a version of the graph that incorporates the choices made in previous iterations), we can combine all of them to recover the entire random spanning tree in the end. 
 
Finally, to make the above approach algorithmically efficient, we utilize a diverse toolkit of modern tools such as Laplacian system solvers, fast approximation algorithms for cut problems, and efficient metric embedding results.

%
%

\subsection{Organization of the Paper}
We begin the technical part of the paper in Section \ref{sec:preliminaries}, where we present some preliminaries on the tools and notions we will need later. Next, in Section \ref{sec:overview}, we provide an overview of our algorithm and the description of its main components. We develop these components in subsequent sections. Specifically, in Section \ref{sec:parititioning}, we describe our graph-cutting procedure and, in particular, establish the connection between the effective resistance metric and graph cuts. Then, in Section \ref{sec:sampling}, we conclude by providing and analyzing our sampling procedure that generates the random spanning tree.

\section{Preliminaries} \label{sec:preliminaries}

Throughout the paper, $G = (V,E)$ denotes an undirected connected graph with vertex set $V$ and edge set $E$ (we allow parallel edges here). Let $V', V''\subseteq V$ be two subsets of vertices of $G$. We define $E(V',V'')$ to be the set of all the edges in $G$ whose one endpoint is in $V'$ and the other one is in $V''$. Also, we define the {\em interior $E(V')$} of $V'$ to be the set of all the edges in $G$ with both their endpoints in $V'$, i.e., $E(V'):=E(V',V')$. Further, the {\em boundary $\partial V'$} of the set $V'$ is the set of all the edges in $G$ with exactly one endpoint in $V'$. We say that an edge $e\in E$ is a {\em boundary edge} of $V'$ if $e\in \partial V'$. Finally, for a collection $\cC$ of subsets of $V$, we say that an edge $e \in E$ is a boundary edge of $\cC$ iff $e$ is a boundary edge of some set $V' \in \cC$.

\subsection{Random Spanning Trees and Random Walks}

Our objective is to sample a random spanning tree $T \subseteq G$ from the uniform distribution, that is, so that for every possible spanning tree $T$ of $G$, the probability that we output $T$ is $\frac{1}{|\cT_G|}$, where $\cT_G$ is the set of all spanning trees of $G$.

In light of the connection that Theorem~\ref{thm:rand_tree_via_rand_walk} establishes between random spanning trees and random walks, the latter notion will be at the center of our considerations.

By a random walk on $G$ we mean the stochastic process corresponding to a walk that traverses the vertices of the graph by moving, in each step, from the current vertex $u$ to one of its neighbors $v$, chosen uniformly at random, over the edge $(u,v) \in E$. It is well-known that this stochastic process has a stationary distribution described by the formula $\pi(v)=\frac{d(v)}{2|E|}$, where $d(v)$ is the degree of $v$ (see, e.g., \cite{Lovasz93}).

Now, by the {\em cover time} $Cov(G)$ of a graph $G$ we mean the expected time it takes for a random walk started at a vertex $s \in V$ to visit (cover) every vertex of $G$. (We take a worst-case bound here over all the possible choices of $s$.)

\subsection{Effective Resistances and the Effective Resistance Metric}
\label{sec:effective_resistance}

Another key notion we will use to establish our result is that of effective resistance and the effective resistance metric. This notion arose originally in the context of study of electrical circuits in physics, but since then it has found a number of equivalent characterizations that highlight its intimate connection to the behavior of random walks and also to random spanning trees themselves (see \cite{DoyleS84,Bollobas98,LyonsP13,Lovasz93}). The characterization that will be especially useful to us was proved in \cite{ChandraRRS89} (see also, e.g., \cite{Lovasz93}). It states that the {\em effective resistance} distance $\Reff(u,v)$ between two points $u$ and $v$ in $G$ is equal to $\frac{\kappa(u,v)}{2 |E|}$, where $\kappa(u,v)$ is the {\em commute time} between $u$ and $v$, i.e., the expected number of steps that a random walk started at $u$ takes before visiting $v$ and returning back to $u$.  Also, even more interestingly in our context, it is known that for any edge $(u,v)$ of $G$, the probability that this edge is included in a random spanning tree is exactly $\Reff(u,v)$. 

Given the above characterization in terms of the commute time, it is not hard to see that the distances $\Reff(u,v)$ give rise to a metric -- called the {\em effective resistance metric}. We may thus define the {\em effective resistance diameter} of a subset of vertices $V' \subseteq V$ to be $\diameff(V') := \max_{x,y \in V'} \Reff(x,y)$. We also set $\diameff(G):= \diameff(V(G))$.


\subsection{Low-Dimensional Embedding of the Effective Resistance Metric} \label{sec:reff_embedding}

In our algorithm we will often be interested in computing the effective resistance distance between pairs of vertices. It is well-known that, for any two vertices $u$ and $v$, one can obtain a very good estimate of the effective resistance $\Reff(u,v)$ in $\tO(m)$ time using the fast (approximate) Laplacian system solvers \cite{SpielmanT03,SpielmanT04,KoutisMP10,KoutisMP11,KelnerOSZ13}. Unfortunately, as we will be interested in distances between many different pairs of points, this running time bound is still prohibitive. However, Spielman and Srivastava \cite{SpielmanS08} designed a way for computing (approximate) effective resistances distances that has significantly better per-query performance. 
\begin{theorem}
\label{thm:SpielmanSrivastava}
Let $G=(V,E)$ be a graph with $m$ edges. For every $\varepsilon>0$ we can find in $\tO(m)$ time a low-dimensional embedding $\cR$ of the effective resistance metric into $\mathbb{R}^{O(\varepsilon^{-2}\log m)}$ such that with high probability
\begin{equation}
\label{eq:SpielmanSrivastava}
\forall u,v\in V \quad (1 - \varepsilon) \Reff(u,v) \le \cR(u,v) \le (1 + \varepsilon) \Reff(u,v),
\end{equation}
and each value of $\cR(u,v)$ can be computed in $O(\varepsilon^{-2}\log m)$ time.
\end{theorem} 
We note that the way we employ the above theorem in our algorithm ensures that even when the guarantee \eqref{eq:SpielmanSrivastava} fails, this does not affect the algorithm's correctness -- just the running time. So, as this guarantee holds with high probability, we can assume from now on that \eqref{eq:SpielmanSrivastava} always holds.

\subsection{Bounding the Cover Time via Effective Resistance}
\label{sec:cover_time_and_effective_resistance}

In Section \ref{sec:effective_resistance}, we have seen the tight connection between effective resistance and the commute time of random walks. As it turns out, drawing on this connection we can relate effective resistance to another key characteristic of random walks: the cover time. To describe this, let us recall the classic results of \cite{AleliunasKLLR79} (see also \cite{Lovasz93}) that the cover time $Cov(G)$ of a (connected) graph with $m$ edges can be bounded from above by $O(mn)$ as well as by $\tO(m\diamdist(G))$, where $\diamdist(G)$ is the graph-distance diameter.

It turns out that by working with the effective resistance-based -- instead of the more traditional graph distance-based -- notion of diameter, one is able to obtain a much tighter characterization of the cover time by replacing $\diamdist(G)$ by $\diameff(G)$ in the bound. It is tighter, as it is not hard to see that the effective resistance distance is always upper-bounded by the graph distance and thus, in particular, $\diameff(G)\leq\diamdist(G)$ for any $G$. Moreover, it is also the case that the resulting upper bound on the cover time is close to being a lower bound. The following lemma, proved in Appendix~\ref{app:cover_weighted}, demonstrates both bounds. It stems from the work of Matthews \cite{Matthews88}. (Also, see \cite{KahnKLV00,DingLP12} for even tighter characterizations of the cover time.) 

\begin{lemma}\label{lem:cover_weighted}
For any graph $G$ we have $m \diameff(G)\leq Cov(G)\leq O(m \diameff(G)\log m)$.
\end{lemma}

In Section~\ref{sec:random_walks_restricted_to_subgraphs} we generalize the above lemma to deal with cover times of {\em subgraphs}. This may be of independent interest.

\subsection{Ball-growing Partitioning} \label{sec:ball_growing}

An important primitive in our algorithm is the following simple but powerful graph partitioning scheme based on the ball-growing technique of Leighton and Rao \cite{LeightonR99}. It is in fact a simplified version of the $(\phi,\gamma)$-decomposition used in~\cite{KelnerM09}.

Let us first introduce the following definition.
\begin{definition}\label{def:ball_decomp}
We say that a partition of a graph $G$ with $m$ edges into disjoint components $V_1,\ldots,V_k$ is a $(\phi, \gamma)$-decomposition of $G$ iff
\vspace{-7pt}
\begin{enumerate}[(1)]\addtolength{\itemsep}{-.5\baselineskip}
\item the graph distance--based diameter $\diamdist(V_i)$ of each $V_i$ is at most $\gamma$,
\item the total number of boundary edges of $V_1,\ldots, V_k$ is at most $\phi \cdot m$.
\end{enumerate}
\end{definition}
As it turns out -- see Appendix~\ref{app:ball_growing} -- we can always find the following $(\phi,\gamma)$-decomposition fast.
\begin{lemma}\label{lem:ball_growing}
For every $\phi=o(1)$ and every graph $G$ with $m$ edges, we can find a $\left(\phi, \frac{2\log (m+1)}{\phi}\right)$-decomposition in $O(m+n)$ time.
\end{lemma}

\subsection{Approximating Minimum $s$-$t$ Cuts} \label{sec:approximating_minimum_cuts}

One of the key parts of our algorithm relies on being able to quickly identify an approximate minimum $s$-$t$ cut in an undirected graph. To do that we use the recent fast $(1+\varepsilon)$-approximate maximum $s$-$t$ flow results of Sherman \cite{Sherman13}, Kelner et al. \cite{KelnerLOS14} and Peng \cite{Peng14}.


\begin{theorem}
For any graph $G$, any two vertices $s$ and $t$, and any $\varepsilon>0$, one can obtain a $(1+\varepsilon)$-approximation to the minimum $s$-$t$ cut problem in $G$ in $\tO(m\varepsilon^{-2})$ time.
\end{theorem}

\section{The Overview of the Algorithm} \label{sec:overview}

The starting point of our algorithm is the random walk--based approach proposed by Theorem \ref{thm:rand_tree_via_rand_walk} and the Kelner-Mądry technique \cite{KelnerM09} of speeding it up. 

Recall that the bottleneck of the vanilla random walk--based approach is that performing a step-by-step simulation of a covering random walk can take $\Omega(mn)$ time. To circumvent this problem,  Kelner-Mądry observed that we actually never need to simulate this covering random walk in full. We just need to be able to recover all its first-visit edges $e_v$. As it turns out, the latter can be done more efficiently by considering a certain ``shortcut'' version of that walk and simulating this shortcutting instead.

Roughly speaking, the design of this Kelner-Mądry shortcutting boils down to finding a partition of the graph $G$ into regions that have diameter at most $\gamma$ each while cutting at most a $\phi$-fraction of edges, for some choice of parameters $\gamma$ and $\phi$. (One can obtain such a partition using the standard ball-growing technique of Leighton and Rao \cite{LeightonR99}.) Then, one starts a step-by-step simulation of a covering random walk in $G$. The crucial difference is that whenever some region is fully covered by the walk, one starts shortcutting any subsequent visits to it. More precisely, upon future revisits to such an already-covered region, one never simulates the walk inside it anymore, but instead immediately jumps out of it according to an appropriate exit distribution. As \cite{KelnerM09} shows, these exit distributions can be computed (approximately) using fast Laplacian system solvers \cite{SpielmanT03,SpielmanT04,KoutisMP10,KoutisMP11,KelnerOSZ13}. Further, the resulting overall (expected) time needed to perform the whole shortcut simulation is at most $\tO(m(\gamma + \phi m))=\tO(m^{3/2})$ (with the optimal choice of parameters being $\gamma=\tO(\sqrt{m})$ and $\phi=\tO(1/\sqrt{m})$).

The above technique can be further refined to yield a total running time of $\tO(mn^{1/2})$, which is an improvement for non-sparse graphs (see \cite{KelnerM09} for details). However, $\Omega(n^{3/2})$ time turns out to be a barrier for this approach, even in the sparse-graph regime. 

\begin{figure}[ht]
\centering
\vspace{8pt}
\includegraphics[width=0.5\textwidth]{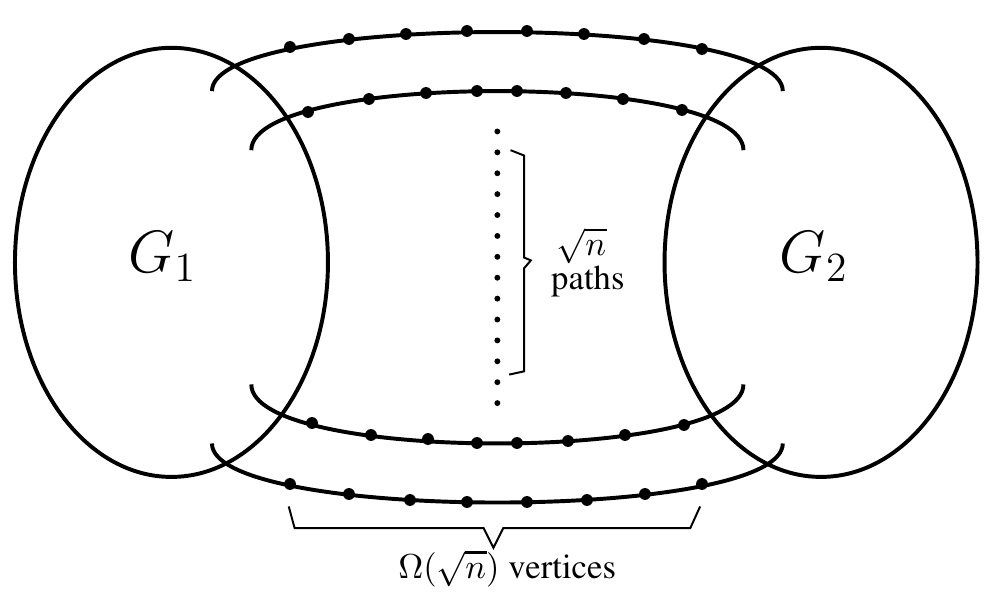}
\vspace{8pt}
\caption{An example of a graph on which Kelner-Mądry algorithm runs in $\Omega(n^{\frac{3}{2}})$ time. It consists of two expanders $G_1$ and $G_2$, each having $\Omega(n)$ vertices, connected by $\sqrt{n}$ disjoint paths of length $\Omega(\sqrt{n})$ each.} 
\label{fig:barrier_example}
\end{figure}

To understand why this is the case, consider the graph depicted in Figure \ref{fig:barrier_example}. This graph has an $O(\sqrt{n})$ diameter. Therefore, if one is interested in generating a random spanning tree only in $\tO(n^{3/2})$ running time, one can simply simulate its covering random walk faithfully -- by a well-known result \cite{AleliunasKLLR79}, this walk will have an expected length of $\tO(n^{3/2})$.

However, if one wanted to use the Kelner-Mądry approach to get an improved running time here, one hits a bottleneck. The difficulty is that the graph cannot be partitioned in a way that makes the diameter $\gamma$ of each piece, as well as the number $\phi m$ of edges cut, both be $o(\sqrt{n})$. (In particular, the two expander-like parts $G_1$ and $G_2$ of this graph are especially problematic due to their size and being far away from each other as well as having large boundaries.) So, the running time of the Kelner-Mądry algorithm is always $\Omega(m(\gamma+\phi m))=\Omega(n^{3/2})$ here.

\subsection{Breaking the $\Omega(n^{3/2})$ Barrier}

Our improved $\tO(m^{4/3})$-time algorithm that, in particular, breaks this $\Omega(n^{3/2})$ running time barrier, stems from a number of new observations. We present them here in the context of the example graph depicted in Figure \ref{fig:barrier_example}. This will serve as an informal motivation for the general techniques that we introduce in the remainder of the paper.

The key observation to make is that random walks are a spectral notion, not a combinatorial one. Therefore, basing our analysis on the effective resistance metric -- instead of applying the traditional graph distance--based optics -- turns out to be a better choice. 

In particular, if one considers our graph in Figure \ref{fig:barrier_example} from this perspective, one can notice that the effective resistance--based geometry of that graph differs from its graph distance--based geometry in a very crucial way. Namely, the effective resistance--based diameter of the whole graph is still $\Omega(n^{1/2})$. (This can be seen by looking at the middle vertices of any two different paths in $G$.) But, the two expander-like parts $G_1$ and $G_2$ -- which were far away from each other in the graph distance--based sense -- are actually very close from the point of view of effective resistance distance! (This difference stems from the fact that the latter notion depends not only on the length of the shortest path but also on the number of disjoint paths.) 

As a result, the union $G_1 \cup G_2$ of these two parts forms a region that can be covered relatively efficiently by a random walk. More precisely, as we show in Lemma~\ref{lem:ultimate}, a random walk in $G$ needs only a small number of steps {inside} $G_1\cup G_2$ before covering this region in full. 

Now, to turn this observation into a fast routine for random spanning tree generation, we need to cope with two issues. Firstly, Lemma~\ref{lem:ultimate} bounds only the number of steps {\em inside} the region $G_1\cup G_2$. Secondly, it measures the number of steps needed before the region is covered -- the number of steps needed to cover the whole graph $G$ might (and actually will) be much larger. 

Dealing with the former issue is relatively straightforward in our example. After all, the outside of our region in $G$ is just a collection of paths. So, it is easy to partition it into small sets with relatively sparse boundaries and then use the shortcutting technique of Kelner-Mądry to simulate the random walk over these sets efficiently. 

To deal with the latter issue, we abandon the original desire of sampling the whole random spanning tree ``in one shot'', i.e., from a trajectory of a single covering random walk. Instead, we first sample only a part of the random spanning tree: its intersection with our region $G_1\cup G_2$. This corresponds to simulating a random walk that covers $G_1\cup G_2$ (but not necessarily the whole graph $G$), and we already know how to do this fast. 

Once this intersection of our region and the random spanning tree is sampled, one can simply encode this sample in $G$ by contracting the edges of our region that were chosen and removing the remaining ones. As $G_1\cup G_2$ constituted a significant fraction of the edges of $G$ (and was also the most troublesome structure in $G$), this sampling makes large progress and thus allows us to simply recurse on the updated (and smaller) version of $G$.  

Now, even though the above discussion was focused on our example graph, it actually already covers almost all the crucial insights behind our general algorithm. 

In particular, at a high level, the main goal of our algorithm is exactly to try to identify in the input graph $G$ some large region that, on one hand, has a small effective resistance diameter and, on the other hand, has an exterior ``shatterable'' enough to enable efficient Kelner-Mądry-type shortcutting. Once such a region is identified, we can make sufficient progress by sampling from its intersection with the random spanning tree, as described above. (Observe that one of the benefits of this approach is that we never need to shortcut the random walk over this region.)

At this point, the only major missing element is a method of dealing with situations when such a large small-radius region with ``shatterable'' exterior does not exist. (For instance, such a situation arises when we remove all but one path connecting $G_1$ and $G_2$ in our example graph.\footnote{In this case, $G_1$ and $G_2$ are no longer close in the effective resistance distance, and neither $G_1$ nor $G_2$ can be partitioned into small sets with sparse boundary.}) To cope with such cases, we establish a relationship between the effective resistance distance and the cut structure of the graph. Specifically, we prove that in such cases there has to exist a small cut that separates two large parts of the graph. (See Lemma~\ref{lem:good_cut} in Section~\ref{sec:parititioning} for details.) Therefore, in lieu of advancing our sampling, we can make progress by refining our grasp on the cut structure of the input graph. 

In the sequel, we proceed to making the above ideas rigorous and introducing the key technical notions behind our algorithm.

\subsection{Subdivisions, Overlays, and Covering Families}\label{sec:covering_families}

The graph-theoretic notions that will be key to our algorithm are subdivisions and covering families. A {\em subdivision} $\cC$ is any collection of { disjoint} subsets of vertices. Then, a {\em covering family} $\ocC=(\cC_0,\ldots, \cC_{\ell})$ is any collection of $\ell+1$ subdivisions $\cC_0,\ldots, \cC_{\ell}$ -- where $\ell := \frac{1}{3} \log m$ --  such that (1) $\cC_0=\{V\}$; and (2) for each $0\leq i\leq \ell$, and any set $C\subseteq V$ in $\cC_i$, the total number of edges in the interior $E(C)$ of $C$ is at most ${m}/{2^i}$, i.e., for each $0\leq i \leq \ell$ and all $C\in \cC_i$,
\begin{equation}
\label{eq:def_size_bound_covering_family}
|E(C)| \leq \frac{m}{2^i},
\end{equation}
where $m$ is the number of edges of our input graph $G$. (Condition (1) in the above definition is just a convention that we employ to ensure that each vertex always belongs to at least one subdivision in the covering family. We also assume $\ell$ to be integral for simplicity.) It is important to observe here that even though sets contained in the same subdivision $\cC_i$ need to be disjoint, sets belonging to different subdivisions can (and often will) overlap.

Now, the key reason for our use of covering families is that, even though they do not directly constitute a partition of the graph $G$, they still provide an implicit -- and much more convenient to deal with from the computational point of view -- representation of certain partitions of $G$. These partitions will play an important role in our algorithm. To define them formally, consider a covering family $\ocC=(\cC_0,\ldots, \cC_{\ell})$. Then, the {\em overlay $\cP(\ocC)$ of $\ocC$} is the partition of $G$ resulting from superimposing all the subdivisions $\cC_i$ of $\ocC$ on top of each other. That is, $\cP(\ocC)$ is a partition of $G$ in which, for any two vertices $v$ and $u$ of $G$, these vertices are in the same component of $\cP(\ocC)$ iff they are not separated by any set in $\ocC$, i.e., for every set $C$ in each subdivision $\cC_i$ of $\ocC$, it is the case that either $u$ and $v$ both belong to $C$ or both are not a part of $C$.

\subsection{Shattering and $\alpha$-bounded Covering Families}\label{sec:shattering_alpha_bounded_families}

Our goal will be to obtain covering families $\ocC$ that, on one hand, have their overlay $\cP(\ocC)$ partition the graph $G$ into components of relatively small effective resistance diameter and, on the other hand, have a relatively small number of boundary edges in each of their subdivisions.

To make this precise, let us fix some covering family $\ocC=(\cC_0,\ldots, \cC_{\ell})$. We will call $\ocC$ {\em shattering} iff all the components in its overlay $\cP(\ocC)$ have effective resistance diameter at most $\Dstar:=56 m^{1/3}$. Also, let us call $\ocC$ {\em $\alpha$-bounded}, for some $\alpha>0$, iff for each $0\leq i\leq \ell$, the number of $i$-based boundary edges in $\ocC$ is at most $\alpha m^{1/3} 2^i$. Here, an edge $e$ is an {\em $i$-based boundary edge} of $\ocC$, for some $0\leq i\leq \ell$, iff $e$ belongs to the boundary $\partial C$ of some set $C\in \cC_i$, but is not part of the boundary of any set $C'$ in $\cC_{i'}$ with $i'<i$.

Before proceeding further, let us observe in passing that the fact that any boundary edge in $\cC_i$ has to be an $i'$-based boundary edge for some $i'\leq i$ makes the $\alpha$-boundedness property also give us a bound on the total number of all boundary edges of each subdivision $\cC_i$. Namely, the following simple fact holds.

\begin{fact}
\label{fa:alpha_boundness_and_boundary_edges}
If $\ocC=(\cC_0,\ldots, \cC_{\ell})$ is an $\alpha$-bounded covering family then, for each $0\leq i\leq \ell$, the total number of boundary edges of the subdivision $\cC_i$ is at most $2\alpha m^{1/3}2^{i}$. Furthermore, the total number of all boundary edges in the whole $\ocC$ is at most $2\alpha m^{1/3}2^{\ell}=O(\alpha m^{2/3})$.
\end{fact}

Now, in the language of the above definitions, our interest will be in finding covering families that are shattering and $\alpha$-bounded for some $\alpha$ being polylogarithmic in $m$.\footnote{The exact value of $\alpha$ is irrelevant for us -- we treat it as a factor hidden under $\tO(\cdot)$.} As we show in Section \ref{sec:parititioning}, such covering families can indeed be found efficiently. In particular, we prove there the following lemma.

\begin{lemma}
\label{lem:partitioning_main}
For any $\alpha>0$ and any $\alpha$-bounded covering family $\ocC=(\cC_0,\ldots, \cC_{\ell})$, we can, in $\tO(m^{4/3})$ time, extend this $\ocC$ -- by only adding new sets to its  subdivisions or subdividing the existing ones -- to a covering family $\ocC'=(\cC_0',\ldots, \cC_{\ell}')$ that is shattering and $\alpha'$-bounded with $\alpha'=\alpha+\alpha^*$ and $\alpha^*$ being a fixed parameter that is polylogarithmic in $m$. 
\end{lemma}

Note that the claimed existence of a $\log^{O(1)}m$-bounded and shattering covering family follows directly from this lemma once one takes $\ocC$ to be a trivial $0$-bounded covering family with $\cC_0=\{V\}$ and all other $\cC_i$ being empty. However, as we will see shortly, our algorithm will actually need the above, more general, statement.

To try to motivate the definition of a covering family and its $\alpha$-boundedness, let us get a little ahead and say that
each boundary edge of $\cC_i$ will incur a preprocessing time cost proportional to the maximum size of a set in $\cC_i$ (cf. Lemma~\ref{lem:cost_shortcutting_sampling}).
For an $\alpha$-bounded covering family $\ocC$, the sum of these quantities is bounded by $\alpha m^{1/3} 2^i \cdot \frac{m}{2^i} = \alpha m^{4/3}$ for every $i$. Thus, we will be able to support our running time bound by maintaining the trade-off between sizes of sets in $\cC_i$ and the number of boundary edges of $\cC_i$.


\subsection{$F$-Marginals, $(F,F')$-Conditioned Graphs, and Minimum-Age Interiors}\label{sec:marginals_overview}

As we mentioned earlier, our algorithm, instead of trying to sample the whole random spanning tree at once, will perform a sequence of samplings from certain marginal distributions. Namely, for a given subset of edges $F\subseteq E$ of the graph $G$, we will be interested in sampling from the marginal distribution -- that we will call the {\em $F$-marginal of $G$} -- describing the intersection of the random spanning tree with the set $F$. 

Observe that the connection between random walks and sampling random spanning trees provided by Theorem \ref{thm:rand_tree_via_rand_walk} can be straightforwardly extended to give a sampling procedure for any $F$-marginal. We can just simulate a covering random walk $X$ in $G$, starting from an arbitrary vertex $s$, and take as our sample the intersection $E_{F}:= \{ e_v : v \in V \setminus \{s\} \} \cap F$ of $F$ and all the first-visit edges $e_v$ to all the vertices of $G$ other than $s$ in that walk. 

Now, the key observation to make here is that to sample from the $F$-marginal of $G$, not only do we not need to simulate such a covering random walk $X$ in full, but also we do not even need to recover all the first-visit edges $e_v$. All we need is just to be able to extract all the edges $e_v$ that could belong to $E_F$. So, as trivially $E_F\subseteq F$, we do not ever have to keep exact track of the movement of the walk $X$ over edges which have no common endpoint with an edge of $F$.

It turns out that having such ability to choose some $F$ and then mostly ignore the parts of the random walks that do not touch $F$ is very convenient and our algorithm takes advantage of that in a crucial way. Namely, as we show below, in any graph, there is always a choice of an appropriate set $F$ so that sampling from the corresponding $F$-marginal can be performed efficiently and it non-trivially advances our progress on recovering the whole random spanning tree.

To formalize this, let us consider some covering family $\ocC=(\cC_0,\ldots, \cC_{\ell})$ and let $\cP(\ocC)$ be the corresponding overlay. Given some component $P\in \cP(\ocC)$ of that overlay, we say that $P$ has {\em age} equal to $r$, for some $0\leq r\leq \ell$, -- we will denote this by $\age(P)=r$ -- iff
\begin{itemize}\addtolength{\itemsep}{-.5\baselineskip}
\item $P$ is a contained in some set $C$ -- that we will call the {\em age witness $\wit{P}$} of $P$ -- from the subdivision $\cC_r$ of $\ocC$;
\item and $P$ is not contained in any set $C'$ from a subdivision $\cC_{r'}$ with $r'>r$.
\end{itemize}
(Note that due to the definition of the overlay it must be that each component $P$ of that overlay can only be either fully contained in or completely disjoint from each set $C$ in the covering family.) Also, observe that, as in our convention each covering family has $\cC_0=\{V\}$, the age of all components of the overlay is well-defined. 

Finally, given a covering family $\ocC$, we define its {\em age} $r^*$ to be the smallest age of all the components of its overlay $\cP(\ocC)$, and its {\em minimum-age interior} $F^*$ to be the union of interiors $E(P)$ of all the components of its overlay $\cP(\ocC)$ that have such minimum age $r^*$.  

\subsection{Efficient Sampling from Minimum-Age Interiors}

Now, one of the crucial properties of the minimum-age interior $F^*$ of a covering family $\ocC$ is that as long as $\ocC$ is shattering and $\log^{O(1)}m$-bounded, we can efficiently sample from the corresponding $F^*$-marginal. More specifically, in Section \ref{sec:sampling} we establish the following lemma.

\begin{lemma}
\label{lem:conditioning_main}
Given any covering family $\ocC=(\cC_0,\ldots, \cC_{\ell})$ that is shattering and $\alpha$-bounded, for some $\alpha$ being polylogarithmic in $m$, let $F^*$ be its minimum-age interior. We can sample from the corresponding $F^*$-marginal in (expected) $\tO(m^{4/3})$ time.
\end{lemma}

Once we have obtained a sample $F'\subseteq F^*$ from our $F^*$-marginal, as in the lemma above, it is important to ensure that this sample remains consistent with all the subsequent samplings that we perform. A simple way to achieve this is by encoding this random choice of $F'$ into the graph $G$. Namely, given the graph $G$, the subset $F^*\subseteq E$ of its edges and the sample $F'\subseteq F^*$ from the corresponding $F^*$-marginal, let us define the {\em $(F^*,F')$-conditioned graph} $G'$ to be $G$ after we contract in it all the edges in $F'$ and remove all the edges in $F^*\setminus F'$. It is not hard to see that sampling a random spanning tree of $G$ that is consistent with our sample $F'$, i.e., finding a random spanning tree $T$ conditioned on the fact that the intersection of $T$ and $F^*$ is exactly $F'$, boils down to finding a random spanning tree in the $(F^*,F')$-conditioned graph $G'$. So, to ensure consistency across all our samplings from consecutive marginal distributions, we just need to make sure to always work with the corresponding conditioned graph. And it is easy to keep track of all edge contractions so as to be able to recover the whole random spanning tree in the original graph at the end.

The key consequence of sampling from the $F^*$-marginal, with $F^*$ being the minimum-age interior, is that it allows us to make non-trivial progress, as measured by the structure of the covering family that we maintain. That is, either the age of that covering family increases, or the interiors of all the components of its overlay partitioning become empty (which by Fact \ref{fa:alpha_boundness_and_boundary_edges} means that there are few edges left overall). This is made precise in the following lemma, whose proof appears in Appendix~\ref{app:conditioning_progress}.

\begin{lemma}
\label{lem:conditioning_progress}
Let $G$ be a graph and let $\ocC=(\cC_0,\ldots, \cC_{\ell})$ be an age-$r^*$ covering family in $G$ that is $\alpha$-bounded, with its minimum-age interior being $F^*$. Given a sample $F'\subseteq F^*$ from the $F^*$-marginal, we can construct, in $\tO(m)$ time, the corresponding $(F^*,F')$-conditioned graph $G'$, as well as a covering family $\ocC'=(\cC_0',\ldots, \cC_{\ell}')$ in that graph. Furthermore, this $\ocC'$ will be $\alpha$-bounded and, if $r^*<\ell$, then the age of $\cC'$ will be at least $r^*+1$; otherwise, i.e., if $r^*=\ell$, then each edge of $G'$ will be in the boundary $\partial C$ of some set $C$ in $\ocC'$.
\end{lemma}
An important technical convention in this context is that the value of $m$ that we use in the definitions of covering family (see \eqref{eq:def_size_bound_covering_family}) and of $\alpha$-boundedness always refers to the number of edges of the {\em original} input graph, i.e., the one before we start removing and contracting edges. (This, of course, means that the running times that we cite in all lemmas are also always proportional to the number of edges of original graph, not the conditioned graph we might be applying them to.)

\subsection{Our Algorithm}\label{sec:algorithm_wrapup}

Having described above all the key ideas and tools behind our approach, we are finally ready to present our algorithm. In fact, at this point, it boils down to a simple iterative procedure.

Namely, our algorithm maintains a graph $G$ and an $\alpha$-bounded covering family $\ocC$ in that graph, for some $\alpha$ being polylogarithmic in $m$. Initially, $G$ is the original input graph and $\ocC$ is the trivial $0$-bounded and age-$0$ covering family with $\cC_0=\{V\}$ and all other $\cC_i$ being empty. 

In each iteration, given $G$ and $\ocC$, we first apply Lemma \ref{lem:partitioning_main} to $\ocC$ to obtain a covering family $\ocC^1$ that is shattering and $\alpha'$-bounded with $\alpha'=\alpha+\alpha^*$. Next, we use Lemma \ref{lem:conditioning_main} to obtain a sample $F'$ from the $F$-marginal corresponding to the minimum-age interior $F$ of $\ocC^1$ and then employ the procedure from Lemma \ref{lem:conditioning_progress} (with $\ocC=\ocC^1$) to obtain the corresponding $\alpha'$-bounded covering family $\ocC'$ in the $(F,F')$-conditioned graph $G'$. (Note that even though $\ocC^1$ was shattering, this $\ocC'$ might no longer be such, as the removal of edges during construction of $G'$ can significantly increase the effective resistance distances.) 

If it is {\em not} the case at this point that all the edges of $G'$ are boundary edges of some set in $\ocC'$, we just proceed to the next iteration with $\ocC$ being equal to $\ocC'$ and $G$ equal to $G'$. 

Otherwise, i.e., if indeed all the edges of $G'$ are boundary edges for some set in $\ocC'$, we can stop our algorithm and note that this condition together with the $\alpha$-boundedness of $\ocC'$ and Fact \ref{fa:alpha_boundness_and_boundary_edges} implies that $G'$ has at most $O(\alpha m^{2/3})=\tO(m^{2/3})$ edges. (Recall that in our convention $m$ denotes the number of edges of the original input graph.) So, this graph is small enough that we can afford to just simulate the whole covering random walk in it and use Theorem~\ref{thm:rand_tree_via_rand_walk} to recover the corresponding random spanning tree. (Observe that the expected length of this covering walk can be easily bounded in terms of the square of the number of edges of our graph -- see Section~\ref{sec:cover_time_and_effective_resistance} -- which results in a bound of $\tO(m^{2/3}\cdot m^{2/3})=\tO(m^{4/3})$.) 

By combining the edges of this tree with the edges contracted during our repeated constructions of conditioned graphs, we obtain the desired random spanning tree in the original graph. It is not hard to see that the above finishing procedure can be implemented in total $\tO(m^{4/3})$ (expected) time, as needed.

Now, to analyze the iterative part of our algorithm, observe that, by Lemma \ref{lem:partitioning_main}, the only operations used to obtain $\ocC^1$ from $\ocC$ are adding new sets and subdividing the existing ones. As these operations cannot decrease the age of any component in the overlay (but could increase it), we can conclude that the age of $\ocC^1$ is at least the age of $\ocC$. As a result, Lemma \ref{lem:conditioning_progress} ensures that with each iteration of our algorithm, the age of the maintained covering family $\ocC$ increases by at least one. 

Therefore, there can be at most $\ell+1=O(\log m)$ such iterations before the algorithm terminates. This observation implies that, on one hand, our maintained covering family $\ocC$ is indeed always $\alpha$-bounded for some $\alpha$ being polylogarithmic in $m$. (This follows since $\alpha$ is initially $0$ and then, by Lemma \ref{lem:partitioning_main}, each iteration increases it by an additive term of at most $\alpha^*$ that is polylogarithmic.) On the other hand, this observation allows us to conclude that by Lemmas \ref{lem:partitioning_main}, \ref{lem:conditioning_main} and \ref{lem:conditioning_progress}, the total (expected) running time of our algorithm is at most $\tO(m^{4/3})$. 

In light of all the above, Theorem \ref{thm:main} follows.

\section{Cutting the Graph} \label{sec:parititioning}

In this section, we prove Lemma \ref{lem:partitioning_main}. That is, we show how to expand a given $\alpha$-bounded covering family $\ocC$ of the graph $G$ to a covering family $\ocC'$ that is still $\alpha'$-bounded for $\alpha'$ that is not much bigger than $\alpha$ and, furthermore, is shattering, i.e., has its overlay $\cP(\ocC')$ partition $G$ into components of small effective resistance diameter. (Consult Section \ref{sec:overview} for necessary definitions.)

At a very high level, our approach is a simple iterative procedure that, as long as $\ocC$ is not yet shattering, considers one-by-one all the components $P$ of the current overlay partition $\cP(\ocC)$ whose effective resistance diameter may still be too large and applies a fixing routine to them. This fixing routine boils down to first identifying good cuts in $G$ that divide the component $P$ and then adding to $\ocC$ carefully-crafted sets (or subdividing the existing ones) based on these identified cuts. (This results in subdivision of $P$ in the resulting overlay $\cP(\ocC)$ and thus in reduction of its diameter.)

Roughly speaking, the way such good cuts for $P$ are identified is via a combination of two partitioning techniques: one purely combinatorial and the other one relying on a certain new connection between the effective resistance metric and graph cut structure. More precisely, in our fixing routine, we first apply the standard ball-growing primitive, as described in Section \ref{sec:ball_growing}, to $P$. If this results in a partition of $P$ into pieces that have sufficiently small diameter and size (while not cutting too many edges), we use the resulting sets to update $\ocC$ accordingly. On the other hand, if that fails, we show that it must be that $P$ contains two large parts that are far away from each other in effective resistance metric. As we prove then, these parts need to have a relatively small cut separating them in the graph $G$ (see Lemma \ref{lem:good_cut} below) and we (approximately) recover this cut using the recent close-to-linear-time minimum $s$-$t$ cut algorithms -- see Section \ref{sec:approximating_minimum_cuts}. Based on that cut, we update $\ocC$ in a very careful manner. 
  
We proceed now to describing our procedure in more detail.

\subsection{The Graph-Cutting Procedure}

Let us first recall the setting of Lemma~\ref{lem:partitioning_main} more formally. We are given an $\alpha$-bounded covering family $\cC = (\cC_0,\dots,\cC_\ell)$ and we want to expand it so that it would become shattering and still be $(\alpha + \alpha^*)$-bounded with $\alpha^* = \log^{O(1)}m$. In this expanding, we are only allowed to add new sets to $\ocC$ and to subdivide the existing ones.

To facilitate our analysis, we also maintain an additional constraint during our procedure. Namely, for each such subdivision $\cC_i$ with $0\leq i<\ell$, we will only be adding new sets -- never subdividing them. (So, subdivision of existing sets can happen only in the subdivision $\cC_\ell$.) Also, we will never modify subdivision $\cC_0$.

Before we start our procedure, we construct, in nearly-linear time, the $(1+\varepsilon)$-approximate low-dimensional embedding $\cR$ of the effective resistance metric of $G$ -- as described in Theorem \ref{thm:SpielmanSrivastava} in Section~\ref{sec:reff_embedding} -- with $\varepsilon=\frac{1}{2}$. This way, we will be always able to compute quickly (i.e., in $O(\log m)$ time), for any two vertices $u$ and $v$ in $G$, an approximation $\cR(u,v)$ of the effective resistance distance $\Reff(u,v)$ between them such that
\begin{equation} \label{eq:embedding_bound}
\frac 23 \cR(u,v) \le \Reff(u,v) \le 2 \cR(u,v).
\end{equation}

Now, our procedure operates by iterating over all the yet-unseen vertices $u$ of $G$. (Initially, all vertices are unseen.) For each such $u$, we identify the component $P \in \cP(\ocC)$ in the overlay $\cP(\ocC)$ of the current $\ocC$ that contains $u$ and run a cutting procedure -- denoted by $Cut(P)$ -- on it. The objective of this procedure is to subdivide $P$ (by adding new sets to $\cC$, or subdividing existing ones, in order to refine $\cP(\ocC)$) into smaller components, all of effective resistance diameter $O(m^{1/3})$. At the end of $Cut(P)$, we mark off all vertices of $P$ as seen. We now describe this procedure.

To this end, let us fix some $P \in \cP(\ocC)$. We start by applying the ball-growing primitive -- as described by Lemma \ref{lem:ball_growing} in  Section~\ref{sec:ball_growing} -- to it to obtain a $(O\left(\frac{\log m}{m^{1/3}}\right),m^{1/3})$-decomposition of $P$ into pieces of effective resistance diameter at most $m^{1/3}$. (Recall that the graph-diameter distance always upper-bounds the effective resistance diameter -- see Section \ref{sec:cover_time_and_effective_resistance}.) For a given piece $P_i$ in this decomposition, we call it {\em large} if its interior $|E(P_i)|$ is of size greater than $m^{2/3}$; otherwise,  we call it \emph{small}. Let $P_1, \dots, P_k$ (possibly $k=0$) be all the large pieces and let $L = \bigcup_{i=1}^k P_i$ be their union. 

We now choose an arbitrary vertex $v_i\in P_i$, for each such large $P_i$, and then use our low-dimensional embedding $\cR$ to (approximately) compute the effective resistance distances between $v_1$ and all the other $v_i$s. (This can be done in $O(k \log m)=\tO(|P|+|E(P)|)$ time.) 

If it turns out that at least one of these distances is larger than $\frac {27}{2} m^{1/3}$, then we say that the ball-growing procedure {\em failed} and proceed to our alternative, more sophisticated, cutting method -- described in Section \ref{sec:ball_growing_fails} below -- that deals with this case. This method encapsulates the key advantage of our choice to employ the effective resistance-based optics instead of the usual graph distance-based one. (Intuitively, the failure of the ball-growing procedure corresponds to the graph structures -- such as the graph in Figure \ref{fig:barrier_example} -- that give rise to the $\Omega(n^{3/2})$ barrier faced by the Kelner-Mądry method \cite{KelnerM09}.)

On the other hand, if ball-growing procedure does not fail, i.e., all these distances $\cR(v_1,v_i)$ are at most $\frac {27}{2} m^{1/3}$, we are able to deal with the component $P$ relatively easily. To see that, observe first that in this case all the large pieces have to be within a ball of small effective resistance diameter, i.e., we can bound the effective resistance distance $\Reff(u,v)$ between any two vertices $v'\in P_{i'}$ and $v''\in P_{i''}$ as 
\begin{equation}\label{eq:bound_diam_large_pieces}
\Reff(v',v'')\leq \dist(v',v_{i'})+\Reff(v_{i'},v_1)+\Reff(v_1,v_{i''})+\dist(v_{i''},v'')\leq 4\cdot \frac{27}{2} m^{1/3} + 2\cdot m^{1/3} = 56m^{1/3},
\end{equation}
where we used~\eqref{eq:embedding_bound} and the fact that each $P_i$ has a diameter of at most $\diameff(P_i)\leq \diamdist(P_i)\leq m^{1/3}$. In other words, $\diameff(L)\leq 56 m^{1/3}$.

Then, we update the covering family $\ocC$ by just adding all the pieces that are {\em small} (not large!) to the subdivision $\cC_\ell$. By definition, all these small pieces are of size at most $m^{2/3} = \frac{m}{2^{\ell}}$, as needed -- cf. \eqref{eq:def_size_bound_covering_family}. Also, note that if $\age(P)=\ell$ and thus $P$ is already contained in some existing $C\in \cC_{\ell}$, adding these pieces is not valid as they would be contained in $C$. In this case, we just modify $C$ by removing all the pieces from it -- effectively, this corresponds to subdividing $C$ and thus is a valid operation. 

Finally, we note that after the above update, the component $P$ is split into components that correspond to each of the small pieces and the union $L$ of all the large pieces. Each of these components -- in particular, $L$ -- have effective resistance diameter of at most $56 m^{1/3}=\Dstar$ (see \eqref{eq:bound_diam_large_pieces}). So they satisfy the diameter condition needed in the definition of a shattering covering family -- see Section \ref{sec:shattering_alpha_bounded_families} -- and thus we can mark all the vertices of $P$ as seen and our cutting procedure can proceed to the next unseen vertex.

\subsection{Cutting the Component $P$ when Ball-growing Fails}\label{sec:ball_growing_fails}

We describe now how to deal with the case when ball-growing fails, i.e., if we have that for at least one large piece $P_i$ with $2\leq i\leq k$ that was produced by running the ball-growing primitive on $P$, $\cR(v_1,v_i)>\frac {27}{2} m^{1/3}$. Let us assume wlog that $i=2$. 

At this point, we can forget about all the other large and small pieces that our ball-growing procedure found in $P$ -- from now on we focus exclusively on the large pieces $P_1$ and $P_2$. The first observation to make in this situation is that the fact that $\cR(v_1,v_2)>\frac {27}{2} m^{1/3}$ implies that $P_1$ and $P_2$ are relatively far from each other in the effective resistance metric. In particular, by~\eqref{eq:embedding_bound}, the effective resistance distance $\Reff(v_1,v_2)$ between $v_1$ and $v_2$ can be lower-bounded as
\begin{equation} \label{eq:before_good_cut_lemma}
\frac 13 \Reff(v_1,v_2) > \frac {9}{2} \cdot \frac 23 m^{1/3} \ge m^{1/3} + \diamdist(P_1) + \diamdist(P_2) \ge m^{1/3} + \diameff(P_1) + \diameff(P_2),
\end{equation}
where we used the fact that the diameter of $P_1$ and $P_2$ is at most $m^{1/3}$. 

Now, the crucial insight here -- and, in fact, one of the key new ideas behind our improvement -- is that distances in the effective resistance metric between vertices are related to the sizes of minimum cut separating them. In particular, as we show in the lemma below, one can relate the fact that $P_1$ and $P_2$ are far away in the effective resistance metric to existence of a small cut separating these sets in $G$. (This connection can be viewed as a certain analogue -- for the minimum $s$-$t$ cut problem -- of the role Cheeger's inequality \cite{AlonM85,Alon86} and its higher-eigenvalue generalization \cite{KwokLLOT13} play in the context of the sparsest cut problem. In both cases, one is tying certain algebraic properties of the Laplacian -- in the case of Cheeger's inequality, these are its eigenvalues, here it is its full (pseudo-)inverse -- to the quality of the corresponding cut in the graph.)

\begin{lemma} \label{lem:good_cut}
Let $U, W \subseteq V$ be disjoint sets in a graph $G$ having at most $m$ edges and let $u$ (resp., $w$) be a vertex in $U$ (resp., in $W$). There exists a cut $\cut^*$ in $G$ that separates $U$ and $W$ and has size at most
\[
|\cut^*| \leq \sqrt{\frac{m}{\gamma_{u,w}(U,W)}},
\] 
where we define $\gamma_{u,w}(U,W):=\frac 13 \Reff(u,w) - \diameff(U) - \diameff(W)$ and assume it to be strictly positive.
\end{lemma}

\begin{proof}
Let us define $k:=\lfloor\sqrt{\frac{m}{\gamma_{u,w}(U,W)}}+1\rfloor$. We claim in the lemma that there is a cut of size at most $k-1$ that separates $U$ from $W$ in $G$. Assume for the sake of contradiction that there is no such cut of size less than $k$. This means that in a graph $G'$ which is created from $G$ by contracting $U$ into a vertex $u'$ and $W$ into a vertex $w'$ there exists no $u'$-$w'$-cut of that size too. As a result, there exist $k$ edge-disjoint paths between $u'$ and $w'$ in $G'$. Once we lift these paths to $G$, we will obtain $k$ paths $p_1,\ldots, p_k$, with each $p_i$ connecting a vertex $u_i \in U$ to a vertex $w_i \in W$ and being of some length $d_i$.

Our goal is to derive a contradiction by showing that $\frac{1}{3}\Reff(u,w)\leq \diameff(U) + \frac{m}{k^2} + \diameff(W)$. (Note that once we use the definition of $k$ and rearrange the terms, this indeed will be a contradiction.) 

To this end, we will use an alternative characterization of effective resistance. According to this characterization, for any two vertices $v$ and $v'$, the effective resistance distance $\Reff(v,v')$ between these vertices is equal to the energy $\energy(f^*)$ of the minimum-energy flow $f^*$ that sends one unit of flow from $v$ to $v'$ in $G$. Here, the energy $\energy(f)$ of a flow $f$ is equal to $\sum_{e} f_e^2$ and corresponds to $G$ being treated as an electric circuit with all edges being unit resistors. (See, e.g., \cite{LyonsP13,Bollobas98} for a proof of this equivalence.)

In light of the above characterization, we can prove upper bounds on $\Reff(u,w)$ by bounding the energy of some (not necessarily optimal) flow $\of$ that pushes one unit of flow from $u$ to $w$ in $G$. To specify the flow $\of$ that we want to use for this purpose, let us first define, for every $i$, $f_i^u$ to be a minimum-energy unit flow from $u$ to $u_i$. Note that the energy $\energy(f_i^u)$ of $f_i^u$ can be at most $\diameff(U)$ -- this is so as, by the above characterization, $\diameff(U)$ is the upper bound on the energy of a minimum-energy unit flow between any two vertices in $U$. Similarly, let $f_i^w$ be a minimum-energy unit flow from $w_i$ to $w$ -- its energy is at most $\diameff(W)$. Finally, let $f_i^d$ be the unit flow from $u_i$ to $w_i$ that simply pushes the whole flow via the path $p_i$ (of length $d_i$). Observe that, for each $i$, we have $\energy(f_i^d)=d_i$. 

Now, our desired flow $\of$ is defined as the following convex combination:
\begin{equation*}
\of := \sum_i \frac{1}{k} \left( f_i^u + f_i^d + f_i^w \right).
\end{equation*}
Note that each sum of the flows $f_i^u$, $f_i^d$ and $f_i^w$ gives rise to a unit flow from $u$ to $w$. As a result, their convex combination $\of$ is also such a unit flow, as desired.

To bound the energy of $\of$, we observe that by virtue of edge-disjointness of the paths $p_i$ we have
\begin{equation}\label{eq:bound_on_p_i_energy}
\energy\left(\sum_i \frac{1}{k} f_i^d \right) = \frac{1}{k^2} \energy \left( \sum_i f_i^d \right) = \frac{1}{k^2} \sum_i \energy(f_i^d) = \frac{1}{k^2} \sum_i d_i \le \frac{m}{k^2},
\end{equation}
where we also used the fact that $\sum_i d_i\leq m$, as $G$ has at most $m$ edges.

Next, we note that the energy is a sum of quadratic functions and thus Jensen's inequality implies that for any set of flows $f_i$ and parameters $\alpha_i \ge 0$ we have 
\[
\energy \left( \sum_i \alpha_i f_i \right) \le \left( \sum_i \alpha_i \right) \left( \sum_i \alpha_i \energy(f_i) \right).
\] 
Using this, our energy-based characterization of effective resistance and the bound \eqref{eq:bound_on_p_i_energy}, we obtain
\begin{align*}
\frac 13 \Reff(u,w) \le& \frac 13 \energy(\of) =\frac{1}{3} \energy \left( \sum_i \frac 1k f_i^u + \sum_i \frac 1k f_i^d+ \sum_i \frac 1k f_i^u \right)\\
\le&\  \energy\left(\sum_i \frac 1k f_i^u\right) + \energy\left(\sum_i \frac 1k f_i^d\right) + \energy\left(\sum_i \frac 1k f_i^w\right) \\
\le& \left( \sum_i \frac 1k \energy(f_i^u) \right) + \energy \left( \sum_i \frac 1k f_i^d \right) + \left( \sum_i \frac 1k \energy(f_i^w) \right) \\
\le& \ \diameff(U) + \frac{m}{k^2} + \diameff(W),
\end{align*}
which is the desired contradiction.
\end{proof}

Once we have established the above lemma, we can apply it in our context and notice that, by \eqref{eq:before_good_cut_lemma}, $\gamma_{v_1,v_2}(P_1,P_2)>m^{1/3}$, which gives us that there exists a cut $K^*$ in $G$ that separates $P_1$ from $P_2$ and is of size at most $m^{1/3}$. As a result, by contracting $P_1$ and $P_2$ into single vertices $v_1$ and $v_2$ and then using the close-to-linear time approximate minimum $s$-$t$ cut algorithms -- see Section \ref{sec:approximating_minimum_cuts} -- on the resulting graph, we can recover, in total $\tO(m)$ time, such a cut $\cut$ that separates $P_1$ from $P_2$ in $G$ and is still of size at most $O(m^{1/3})$. (It is worth noting here that the cut $\cut$ is a global cut in $G$, i.e., it is not contained in the component $P$ that we are processing and, in fact, it can intersect all the sets in our covering family $\ocC$.)  

Our goal now is to use this cut $\cut$ to carefully craft a new set $S$ to be added to one of the subdivisions $\cC_i$ for $i<\ell$. Adding this set will ensure that $P_1$ and $P_2$ become separated in the resulting new overlay $\cP(\ocC)$. To achieve this, we need to ensure, in particular, that, on one hand, this set does not intersect any existing sets in $\cC_i$ and, on the other hand, that it does not contain too many boundary edges. As the following lemma (proved in Appendix~\ref{app:massaging}) shows, such a set $S$ can indeed be found and, additionally, it satisfies all the other constraints that we imposed on our algorithm. 

\begin{lemma} \label{lem:massaging}
Let $\ocC = (\cC_0,\ldots,\cC_\ell)$ be a covering family of a graph $G$, and let $P$ be a component of its overlay partition $\cP(\ocC)$. Also, let $P_1$ and $P_2$ be two subsets of $P$ satisfying $|E(P_1)|, |E(P_2)| > m^{2/3}$. Given a cut-set of edges $\cut$ which separates $P_1$ and $P_2$ in $G$, we can, in $\tO(m)$ time, produce a set of vertices $S$ and an $i$ such that:
\vspace{-7pt}
\begin{enumerate}[(1)]\addtolength{\itemsep}{-.5\baselineskip}
	\item $0 < i < \ell$, \label{cond:i_0_ell}
	\item the size of the interior of $S$ satisfies $\frac{m}{2^{i+1}} < |E(S)| \le \frac{m}{2^{i}}$, \label{cond:size}
	\item $S$ is disjoint from all sets in subdivision $\cC_i$, \label{cond:disjoint}
	\item exactly one of the sets $P_1$, $P_2$ is contained in $S$, and the other one is disjoint from $S$, \label{cond:separation}
	\item every boundary edge $e \in \partial S$ is either a boundary edge of $\cC_j$ with $j \le i$ or it comes from the cut $\cut$, i.e., $e \in \cut$. \label{cond:boundary_edges}
\end{enumerate}
\end{lemma}

At this point, we can use Lemma~\ref{lem:massaging} above to produce the set $S$ that we then add to the subdivision $\cC_i$ with $i < \ell$ (this is valid thanks to conditions \eqref{cond:disjoint} and \eqref{cond:size}).
As a result of this update, in the new overlay $\cP(\ocC)$, the component $P$ has been separated into two components $P' = P \cap S$ and $P'' = P \setminus S$ (see condition~\eqref{cond:separation} in the Lemma \ref{lem:massaging}), both of which might still have too high effective resistance diameter. For this reason, we still need to recurse on both these pieces, i.e., we call $Cut(P')$ and $Cut(P'')$. As we will shortly see, even though our whole cutting procedure that added this new set $S$ was computationally very expensive, i.e., it took $\tO(m)$ time, which might be very large compared to the size of $P$, we can justify this cost by proving that we do not need to run this procedure too often. (Roughly speaking, this stems from an observation that the sparsity of the cut that $S$ constitutes is relatively small and thus an appropriate charging argument can be applied.)

\subsection{Analysis of the Cutting Procedure}

We now proceed to the analysis which shows that our cutting procedure fulfills all its promises.

Observe that, from the point of view of the running time analysis, our main concern is that whenever ball-growing fails, the runtime of the cutting procedure handling this situation is as large as $\tO(m)$. Fortunately, as the following lemma shows, this does not happen too often. 

\begin{lemma} \label{lem:number_of_insertions}
The number of times a new set is inserted into $\cC_i$ is less than $2^{i+1}$ for $i < \ell$.
\label{lem:number_of_insertions_of_set}
\end{lemma}

To see how this provides a bound on the total running time of the cutting procedure from Section \ref{sec:ball_growing_fails}, recall that these expensive cut computations are always followed by an insertion of a new set into $\cC_i$ with $i < \ell$. By the above lemma, summing over all $0<i<\ell$, we have a total of at most $2^{\ell}=O(m^{1/3})$ of such insertions. The resulting running time bound is $\tO(m^{4/3})$, as desired.

\begin{proof}[of Lemma~\ref{lem:number_of_insertions}]
Note that a new set $S$ may be inserted into $\cC_i$ with $i < \ell$ only when ball-growing fails, i.e., only during the cutting procedure described above in Section \ref{sec:ball_growing_fails}. By condition~\eqref{cond:size} of Lemma~\ref{lem:massaging}, the size of $S$ satisfies $\frac{m}{2^{i+1}} < |E(S)|$. Moreover, by condition \eqref{cond:disjoint}, it is disjoint from all other sets in $\cC_i$, thus any edge of $G$ can be covered by a set from $\cC_i$ only once. So, as the insertion of one set $S$ covers more than $\frac{m}{2^{i+1}}$ new edges, there must be less than $2^{i+1}$ such insertions.
\end{proof}

The rest of the running time analysis and the adequacy of our cutting procedure is summarized in the following two statements, whose proofs may be found in Appendices \ref{app:partitioning_running_time} and \ref{app:partitioning_alpha_boundedness}, respectively.

\begin{lemma} \label{lem:partitioning_running_time}
Our whole cutting procedure runs in $\tO(m^{4/3 })$ time.
\end{lemma}

\begin{lemma} \label{lem:partitioning_alpha_boundedness}
The new covering family $\cC$ is $(\alpha + \alpha^*)$-bounded with $\alpha^* = \log^{O(1)}m$.
\end{lemma}

We are now in a better position to understand the idea behind the notion of $\alpha$-bounded covering families. Namely, if we have a set $C$ that we want to insert into our covering family $\ocC$, we should strive to do so at the lowest level possible (i.e., the lowest level $i$ where the condition $|E(C)| \le \frac{m}{2^i}$ is satisfied). By doing so, we will guarantee the $\alpha$-boundedness of the resulting covering family (as above), which encodes the trade-off between the number of boundary edges on any level $i$ and sizes of sets on that level (see the discussion in Section~\ref{sec:shattering_alpha_bounded_families}). For instance, if we have a large set $C$, we should place it on a low level $i$ of the covering family; this, through the $\alpha$-boundedness property, guarantees that there will be few boundary edges on that level (and such edges are expensive, as each one incurs a preprocessing time cost proportional to the maximum size of sets in $\cC_i$ -- which is roughly the size of $C$). On the other hand, we should place small sets on high levels of the family; because they are small, we can afford to have many more boundary edges incident to them.
In this way, one can view the $O(\log m)$-level covering family as a way of making sure that we have a right level $i$ on which to insert any set $C$ that we might want to use for shortcutting our random walk.

\section{Fast Sampling from Marginal Distributions} \label{sec:sampling} 

In this section, we prove Lemma \ref{lem:conditioning_main}. That is, we show how, given a covering family $\ocC=(\cC_0,\ldots \cC_{\ell})$ in a graph $G$ that is shattering and $\alpha$-bounded, for some $\alpha$ being polylogarithmic in $m$, one can sample from the minimum-age interior $F^*$ of $\ocC$ in (expected) $\tO(m^{4/3})$ time. (See Section \ref{sec:overview} for all the necessary definitions, and recall that in our convention $m$ refers to the number of edges of our original input graph and thus our graph $G$ here can have potentially much less than $m$ edges.)

As the covering family $\ocC$ that we will be dealing here is fixed, let us simply use $\cP$ to denote the overlay partition $\cP(\ocC)$ of $\ocC$ and fix $r^*$ to be the age of $\ocC$, i.e., the age $\age(P)$ of a minimum-age component $P$ in $\cP$. Also, let us denote by $\cP^*$ the collection of all such minimum-age components $P$ of $\cP$, and make $V^*$ be the union of the vertices of all the components in $\cP^*$. 

Our goal is to sample from the $F^*$-marginal distribution of $G$ (cf.~Section \ref{sec:marginals_overview}). This means that we want to sample a random subset of $F^*$ according to the distribution that returns $F\subseteq F^*$ with probability 
\[
|\{T\in \cT_G: T\cap F^*=F\}|/|\cT_G|,
\]
where $\cT_G$ denotes the set of all spanning trees of $G$. In other words, we want to find the intersection of $F^*$ with a random spanning tree.

As we already discussed, in light of Theorem \ref{thm:rand_tree_via_rand_walk}, this could be done by, first, simulating the random walk starting at some vertex $s$ until all the first-visit edges $e_v$ to $v\in V^*\setminus \{s\}$ are known and, then, returning the intersection $F'$ of $F^*$ and all the edges $e_v$ found. 

Unfortunately, even though we do not need to learn here what the first-visit edges $e_v$ are for all vertices of $G$ -- only the ones that could be in $F^*$, the expected length of the necessary random walk might still be comparable to the cover time of the whole graph. Therefore, we cannot afford to simply execute this simulation.

\subsection{Shortcutting}\label{sec:shortcutting}
To cope with the above-mentioned problem and obtain the desired efficient sampling procedure, we build on the approach developed by Kelner and Mądry \cite{KelnerM09} to sample from the whole random spanning tree distribution. 

Namely, the key observation to make here is the following: even though the transcript of the whole walk defining the first-edges $e_v$ that we need, might be quite long, i.e., even $\Omega(mn)$, in the end, we only utilize at most $|V^*|\leq n$ of its entries. Therefore, one might hope that there is a way of directly generating a  shortcut transcript of that walk, that is much shorter (and, hopefully, possible to generate much more efficiently) while still retaining the information we need.

To get an intuition how that shortcutting could be implemented, let us consider some subset $S$ of vertices of $G$. (For now, one can think that $S$ corresponds to one of the components of $\cP^*$.) Imagine that we are faithfully simulating the random walk in $G$ step-by-step, except that once we cover $S$ we start shortcutting our trajectories inside this set. Namely, whenever we enter some vertex $v\in S$ after $S$ was already covered, we would like to immediately leave $S$ by jumping out over one of the edges of its boundary $\partial S$. Observe that by doing so we are not losing any information that we need later -- at this point, we already know all the first-visit edges $e_v$ for all the vertices in $S$. So, the only challenge here is to ensure that the edge through which we leave $S$ is sampled consistently with the distribution induced by the random walk.  

More precisely, to implement the above shortcutting strategy, we need to be able, for each vertex $v\in S$, to sample from the distribution $\{P_v(e)\}_{e\in \partial S}$, where $P_v(e)$ is the probability that a random walk started at $v$ exits $S$ for the first time via the edge $e$.

Fortunately, as shown by Kelner and Mądry \cite{KelnerM09} and Propp \cite{Propp10}, by exploiting the connection between random walks, electrical flows and the fast (approximate) Laplacian system solvers \cite{SpielmanT03,SpielmanT04,KoutisMP10,KoutisMP11,KelnerOSZ13}, one can indeed construct a data structure -- from now on, called {\em $S$-shortcut sampler} -- that allows one to quickly sample from such edge-exit distribution. Moreover, as made precise in the lemma below, this construction is efficient provided that the size of the boundary of $S$ is relatively small compared to the size of its interior. 

\begin{lemma}\label{lem:cost_shortcutting_sampling}
Given a subset $S \subset V$, we can construct a data structure -- called {\em $S$-shortcut sampler} -- that allows one to sample, for a given vertex $v\in S$, an edge $e\in \partial S$ according to the distribution $\{P_v(e)\}_{e\in \partial S}$ in $O(\log m)$ time. The total expected time spent on constructing and maintaining this data structure is $\tO(|S|+|E(S)||\partial S|)$.
\end{lemma}
This lemma follows from a simple adaptation of the techniques of Kelner and Mądry (see Lemma 7 and 9 in \cite{KelnerM09}) and the observation of Propp \cite{Propp10} (also see Lemma~7.3.2 in \cite{Madry11}). We sketch its proof in Appendix~\ref{app:cost_shortcutting_sampling}. An interesting feature of this construction is that it never computes the exact values of the probabilities $P_v(e)$ (this is a consequence of fast Laplacian system solvers being only approximate). Nonetheless, it still allows one to obtain an {\em exact} sample from the corresponding distribution.

\subsection{Our Sampling Algorithm}

Having introduced the shortcutting technique, we can proceed to describing our sampling algorithm. For now, we do it assuming a (technical) condition that the overall effective resistance diameter $\diameff(G)$ of our graph is relatively small, i.e., $\diameff(G) = O(m^{2/3})$. Later, in Section \ref{sec:deadling_with_large_diameter_case}, we deal with this assumption. 

Recall that in \cite{KelnerM09}, Kelner and Mądry obtain their shortcut transcript of a covering random walk in $G$ by, first, finding a simple ball-growing-based partition of the graph into pieces that have small graph diameter and, then, running simulation of the walk in $G$ while implementing the shortcutting procedure (as described above) for each such piece. As this partition cuts relatively small number of edges and -- due to their small diameter -- the individual cover time of each piece is small too, they are able to recover the whole random spanning tree within their claimed running time bounds. 

Given all of that, it might be tempting to follow this approach also in our case. Namely, we could simulate a random walk in $G$ until it covers all the vertices in $V^*$, while applying the shortcutting procedure to each component of the overlay partition $\cP$ of our covering family $\ocC$. In this case, we know that all these components have small effective resistance diameter (as $\ocC$ is shattering). So, it follows (cf. Section~\ref{sec:random_walks_restricted_to_subgraphs}) that the individual cover time of each is sufficiently small, as needed. Furthermore, one might hope that the fact that $\ocC$ is $\alpha$-bounded implies that the size of boundaries of each of these components is also sufficiently small (compared to their size) and thus the construction times of corresponding $S$-shortcut samplers (see Lemma~\ref{lem:cost_shortcutting_sampling}) are acceptable. 

Unfortunately, in general, this is not the case. It might actually happen that the sizes of boundaries of the components of $\cP$ are prohibitively large and thus we are still hitting here the same bottleneck as the original Kelner-Mądry approach faces -- see our discussion in Section \ref{sec:our_approach}. (In fact, the graph from Figure \ref{fig:barrier_example} is exactly the example showing that.)

In light of the above, we have to design our shortcutting strategy in a different and much more careful way. This will crucially exploit our lack of need of knowing what the first-visit edges are for the components of $\cP$ that are not in $\cP^*$. Namely, the key realization here is that if all we care about is preserving the first-visit information only for the components of $\cP^*$, there is much more flexibility in the way we implement the shortcutting of our walk. In particular, for a given component $P$ of $\cP^*$ we do not really need to take $S=P$ and construct the corresponding $S$-shortcut sampler (cf. Lemma~\ref{lem:cost_shortcutting_sampling}); it suffices to find some different -- {\em larger} set $S(P)$ that, on one hand, contains $P$ and has small number of boundary edges and, on the other hand, does not overlap any other component in $\cP^*$ (but still can overlap components that are components of $\cP\setminus \cP^*$). In fact, a similar observation applies also to other components of $\cP$ (that might not be in $\cP^*$).

To make this precise, for a given component $P$ of the overlay $\cP$ (that might or might not be in $\cP^*$), let us call a set $S\subseteq V$ an {\em escape set} for $P$ iff $P\subseteq S$ and $S$ does not intersect any component of $\cP^*$ other than $P$ (if $P$ is in $\cP^*$).

Now, it turns out that once we find for each component $P$ its escape set $S(P)$ which is not too big and has sufficiently small boundary, we can use these sets to generate a shortcut version of our random walk. This shortcutting is based on applying a procedure that, whenever our walk enters a vertex $v$ belonging to some component $P$ of $\cP$, does the following:
\begin{itemize}
	\item if $P \in \cP^*$:
	\begin{itemize}
		\item if $P$ has already been covered, perform a jump out of $S(P)$;
		\item otherwise: faithfully simulate a single step of the walk;
	\end{itemize}
	\item otherwise: perform a jump out of $S(P)$.
\end{itemize}
We stop the random walk at the moment we have covered all components $P \in \cP^*$. Note that this procedure always faithfully simulates the walk inside a component from $\cP^*$ until this component is fully covered. Also, whenever we perfom a jump over an escape set and land directly inside a component of $\cP*$, we know what the last traversed edge (through which we entered this component) was. So this shortcutting still retains all the information needed to recover the sample from the $F^*$-marginal.

Now, of course, the crucial part is choosing the sets $S(P)$ appropriately, so that the resulting running time of our simulation is within the desired bounds. This is the aim of the next section. 

\subsection{Choosing Appropriate Escape Sets}
Our goal here is to define the escape sets $S(P)$ for $P\in \cP$ that we use for shortcutting the walk, as described above. To this end, let us introduce a new piece of notation.

For a given $0\leq j\leq \ell$, let us define $\cP_j$ to be the partition resulting from superimposing on top of each other all the sets from subdivisions $\cC_0, \cC_1, \dots, \cC_j$. (So, we have that $\cP_{\ell} = \cP$, $\cP_0 = \{V\}$ and each $\cP_{j}$ is a coarsening of $\cP_{j+1}$.)

 Furthermore, for a component $P \in \cP$, let us define $\cP_j(P)$ to be the component of $\cP_j$ that contains $P$. Finally, we say that a component $R$ of $\cP_j$ is \emph{modest} if it is contained in a set $S \in \cC_j$. The reason why we want to distinguish modest components is because they are small enough for our purposes. Namely, if $R$ is modest on level $j$ and $R\subseteq S\in \cC_j$, then $|E(R)|\leq |E(S)|\leq m/2^{j}$, which might not be true for an arbitrary component.

\begin{lemma}
\label{lem:different_Pp_components}
Let $P_1$ and $P_2$ be different components of $\cP$ with $\age(P_1) = \age(P_2) = r$. Then they belong to different components of $\cP_r$: $\cP_r(P_1) \neq \cP_r(P_2)$.
\end{lemma}
\begin{proof}
Since $P_1 \neq P_2$, for some $i^*$ there must exist a set $S \in \cC_{i^*}$ such that exactly one of $P_1,P_2$ is contained in $S$. But because of age, neither $P_1$ nor $P_2$ is contained in any set $S \in \cC_i$ for $i > r$. Therefore, $i^* \le r$.
\end{proof}

We infer from the lemma that $\cP_{r^*}(P_1)$ and $\cP_{r^*}(P_2)$ are disjoint and are both modest whenever $P_1,P_2$ are two different pieces in $\cP^*$. It follows that $\cP_{r^*}(P)$ satisfies all the required conditions for an escape set of $P\in \cP^*$, hence we define $S(P):=\cP_{r^*}(P)$. 

For the remaining components $P\in \cP\setminus \cP^*$ we choose $S(P):=W(P)$ -- the witness set of $P$. Then we have $P\subseteq W(P)$, $W(P)\in \cC_{\age(P)}$ and $\age(P)>r^*$. So clearly, none of the components in $\cP^*$ intersects $S(P)$.

Let us now look more closely at the exceptional situation when $r^*=0$. One may see that in this case there is exactly one component in $\cP$ of age $0$, say $P$. It appears that according to our definition, $S(P)=V$, which makes little sense (we cannot jump out of $V$). Fortunately, in this special case we stop the walk immediately after the single component gets covered. This means there is no use for the $S(P)$ set. Note that this issue concerns only the border case $r^*=0$; if $r^*>0$, every $S(P)$ is a proper subset of $V$ and has some boundary edges through which we may jump out of it.

In Section~\ref{sec:analysis_sampling} we will see that such a choice of escape sets allows us to prepare all $S(P)$-shortcut samplers in reasonable time. Now we wish to proceed to the runtime analysis of our entire sampling procedure. We start with some preparations in the next section.

\subsection{Random Walks Restricted To Subgraphs} \label{sec:random_walks_restricted_to_subgraphs}

In this section we state a lemma which will be instrumental to our analysis of the random walk simulation. It bounds the expected length of a random walk which covers a given subgraph.

\begin{lemma} \label{lem:ultimate}
Let $G = (V,E)$ be an unweighted graph and $V' \subseteq V$. Also let $E' \subseteq E(V')$ be any set of edges inside $V'$. Suppose we run a random walk in $G$ from an arbitrary vertex (not necessarily from $V'$) and stop it once all vertices from $V'$ have been visited. Then the expected number of traversals of edges from the set $E'$ is \[O(|E'| \cdot \diameff(V') \cdot \log n).\]
\end{lemma}

We prove this lemma in Appendix~\ref{app:proof_of_ultimate}. It may be viewed as a generalization of Lemma~\ref{lem:cover_weighted} in two aspects. First, it lets us look only at traversals of a subset of edges of the graph (for example, in the analysis of our algorithm this might be the set of boundary edges of $\ocC$). Second, we can restrict our attention to a subgraph induced by a vertex subset $V'$ and answer the question: what is the {\em subgraph cover time}, i.e., how many traversals of edges between vertices of $V'$ will there be if we decide to stop the walk as soon as $V'$ is covered? 

Note that this improves upon Lemma~6 of~\cite{KelnerM09}, not only by considering $\diameff$ rather than $\diamdist$, but also by dropping the assumption that the number of boundary edges of $V'$ is at most $|E(V')|$, and by letting us only consider a subset of edges. It also subsumes Fact~5 of~\cite{KelnerM09}. We thus see that both directions of generalizing the upper bound of $O(mn)$ for the cover time were known previously, but in weaker versions, separately and not using effective resistance. We believe that Lemma~\ref{lem:ultimate} is a robust and useful statement which is of independent interest.

We now apply the above lemma to analyze our algorithm.

\subsection{Analysis of the Sampling Algorithm}\label{sec:analysis_sampling}

If we recall our simulation of the random walk that we described in the previous section, there are two main ingredients in the running time of that procedure: the expected time needed to prepare the data structures for implementing shortcutting as in Section \ref{sec:shortcutting}, and the expected running time of the actual simulation of the shortcut random walk. Below, we analyze both of these separately.

\begin{lemma}\label{lem:preparation_cost}
The total expected time spent on preparing the structures for shortcutting sets $S(P)$ for $P\in \cP^*$ is $\tO(m^{4/3})$.
\end{lemma}
\begin{proof}
To initialize our data structures, we need to construct the shortcut samplers for all components in $\cC_i$ for $i=1,\dots,\ell$ and modest components of $\cP_{r^*}$. By Lemma \ref{lem:cost_shortcutting_sampling}, computing and maintaining a shortcut sampler for a set $S$ can be done in time $\tO(|S|+|E(S)|\cdot |\partial S|)$, where $\partial S$ is the set of boundary edges of $S$.

Consider first all sets $S\in \cC_i$, for some $i$. By Fact~\ref{fa:alpha_boundness_and_boundary_edges} and the bound $|E(S)| \le \frac{m}{2^{i}}$, the time needed to maintain the corresponding data structures is at most
\begin{eqnarray*}
\sum_{S \in \cC_i}\tO(|S|+|E(S)|\cdot |\partial S|)&=&\tO\left(n+\sum_{S \in \cC_i}|E(S)|\cdot |\partial S|\right)=\tO\left(n+\frac{m}{2^{i}} \cdot \sum_{S \in \cC_i} |\partial S|\right)\\
&=&\tO\left(n+\frac{m}{2^{i}} \cdot \sum_{S \in \cC_i} |\partial S|\right)=\tO\left(n+\frac{m}{2^{i}} \cdot m^{1/3}2^{i}\right)=\tO(m^{4/3}).
\end{eqnarray*}
Then, summing this time over all $i=1,\ldots,\ell$, we obtain that the overall time is $\tO(m^{4/3})$ as well.

Secondly, let us consider all modest components $R \in \cP_{r^*}$. First notice that $e \in E$ is a boundary edge of $\cP_{r^*}$ iff it is a boundary edge of $\cC_j$ for some $i \le r^*$. Hence, the number of such edges is, by Fact~\ref{fa:alpha_boundness_and_boundary_edges},
\[
\sum_{i \le r^*} \tO\left( m^{1/3} \cdot 2^i \right) = \tO(m^{1/3} \cdot 2^{r^*})
\]
as well. Also, for a modest component $R$ there exists an $S \in \cC_{r^*}$ with $R \subseteq S$, which implies that $|E(R)| \le |E(S)| \le \frac{m}{2^{r^*}}$. As a result, we can see that the previous calculation (that we made for the sets $S\in \cC_i$) holds also in this case. We conclude that the total time spent on maintaining shortcut samplers is $\tO(m^{4/3})$, as desired.
\end{proof}

\begin{lemma}\label{lem:walking_cost}
The expected time of performing the shortcut random walk simulation is $\tO(m^{4/3})$.
\end{lemma}

\begin{proof}
Observe that there are two types of steps in our procedure:
\begin{enumerate}[(a)]
\item regular, single steps inside some age-$r^*$ component, before it is covered,
\item jumps over a $\cC_j$ set (for some $j>r^*$) or over a component in $\cP_{r^*}$.
\end{enumerate}
Consider a component $P\in \cP^*$. The expected time spent on steps of type (a) within $P$ is bounded by $O(|E(P)| \cdot \diameff(P) \cdot \log n)$ by Lemma~\ref{lem:ultimate} (used with $V' = P$ and $E' = E(P)$). Since the effective resistance diameter of $P$ is $O(m^{1/3})$, this is $\tO(m^{1/3} |E(P)|)$. Thus, the expected total time spent on type-(a) steps is 
\[
\sum_{P\in \cP^*} \tO(m^{1/3} |E(P)|)=\tO(m^{1/3} \sum_{P\in \cP^*} |E(P)|)=\tO(m^{1/3} \cdot m)=\tO(m^{4/3}).
\]

Let us now concentrate on steps of type (b). Recall that every step of this kind is carried out using one query to the corresponding shortcutting structure and thus it takes $O(\log m)$ time. Each such jump ends on a boundary edge (either of some $\cC_j$ set or some component in $\cP_{r^*}$). So we can think that every jump corresponds to a boundary edge traversal in the underlying unshortcut random walk.

 We can now bound the time required for type-(b) steps simply by the expected number of times a boundary edge is traversed during our simulation. Note that by the time the whole graph is covered for the first time we are certain to be done with our simulation. Therefore we may use Lemma~\ref{lem:ultimate} to bound the time spent on type-(b) steps. We set $E'$ to be the set of all boundary edges, $V' = V$, and we recall that $\diameff(G) = O(m^{2/3})$ (the small-diameter assumption) and $|E'| = \tO(m^{2/3})$. Thus we obtain the desired bound of $\tO(m^{4/3})$. 

\end{proof}

By Lemmas~\ref{lem:preparation_cost}~and~\ref{lem:walking_cost} the total expected running time of sampling (in the small-diameter case) is $\tO(m^{4/3})$.

\subsection{Dealing with the Large-Diameter Case}\label{sec:deadling_with_large_diameter_case}

Recall that so far we were assuming for simplicity that the effective resistance diameter of $G$ is small, i.e.,  $\diameff(G)=O(m^{2/3})$. We will show how to deal with this assumption by slightly changing our algorithm. 

We have a graph $G$ and a collection of disjoint subgraphs $\cP^* \subseteq \cP$. We intend to sample a set according to the $F^*$-marginal distribution, where $F^*=\bigcup_{P\in \cP^*}E(P)$. 
We already know how to simulate the random walk in $G$ in such a way that we pay only for:
\begin{enumerate}[(a)]
\item traversing interior edges of every component $P\in \cP^*$ until it is covered,
\item traversing edges coming from a set $E'$ of size $\tO(m^{2/3})$.
\end{enumerate}
As seen in the proof of Lemma~\ref{lem:walking_cost}, the time complexity of (a) is $\tO(m^{4/3})$ irrespective of $G$, but the time complexity of (b) is $\tO(\diameff(G)m^{2/3})$, which is too much if we are not able to bound the diameter of $G$ by $O(m^{2/3})$. To circumvent this problem, we use the ball-growing procedure to cut $G$ into small-diameter pieces (cf.~\ref{sec:ball_growing}). Namely, we obtain a partition $V=\bigcup_{i=1}^{q} F_i$ satisfying $\diamdist(F_i)=O(m^{2/3})$ and $|E''|=\tO(m^{1/3})$, where $E''=E\setminus \bigcup_{i=1}^{q} E(F_i)$ is the set of boundary edges of this partition. Take $E'$ to be the set of all boundary edges (of all subdivisions $\cC_i$). By $\log^{O(1)}m$-boundedness of $\ocC$ we have $|E'| = \tO(m^{2/3})$.

In the new algorithm, whenever we cover all interesting vertices (those from $V^*$) within some $F_i$, we start shortcutting the walk over the whole $F_i$. This allows us to treat every edge $e\in E'\setminus E''$ as part of a small-diameter graph $F_i$ and pay for traversing it only until $F_i$ is covered. In the next subsection we make this intuition formal.

We now describe the strategy used to simulate the random walk. Suppose we are at a vertex $v$ and want to make a step in our simulation. Let $F_i$ be the small-diameter piece that $v$ belongs to. Then:
\begin{enumerate}
	\item if all vertices in $V^*\cap F_i$ have already been covered: jump over $F_i$ in one step,
	\item otherwise:
	\begin{enumerate}
		\item let $P$ be the component in $\cP$ such that $v\in P$,
		\item if $P\in \cP^*$:
		\begin{enumerate}
			\item if $P$ is covered: jump over the $S(P)$ set,
			\item else: perform one regular step of the random walk,
		\end{enumerate}
		\item else: jump over $S(P)$.
	\end{enumerate}
\end{enumerate}
Note that case 2. is exactly the same as our simplified strategy from the previous subsections.

\subsection{Analysis of Sampling in the Large-Diameter Case}
As in the small-diameter case, we need to bound two quantities: time spent on preparing the shortcutting structures and the expected time of simulation.
\begin{lemma}\label{lem:preparation_cost_general}
The total expected time spent on preparing the structures for shortcutting sets $\cC_i$ for $i=1,\dots,\ell$, modest components of $\cP_{r^*}$ and $F_i$ for $i=1,\ldots,q$ is $\tO(m^{4/3})$.
\end{lemma}
\begin{proof}
An additional part compared to the simplified case are the $F_i$s. The expected time needed to maintain all the shortcut samplers for them is (use Lemma~\ref{lem:cost_shortcutting_sampling}): $\sum_{i=1}^{q} \tO(|F_i|+|E(F_i)||\partial F_i|)=\tO(n+\sum_{i=1}^{q}|E(F_i)||\partial F_i|)=\tO(\sum_{i=1}^{q}|E(F_i)||E''|)=\tO(m|E''|)=\tO(m^{4/3})$.
\end{proof}

\begin{lemma}\label{lem:walking_cost_general}
The expected time of performing the random walk simulation is $\tO(m^{4/3})$.
\end{lemma}
\begin{proof}
We remarked previously that the expected time spent on covering $P$'s is bounded by $\tO(m^{4/3})$. Thus we only need to estimate the expected number of jumps: those from step 1. and those from step 2. Like in the small-diameter case, we argue that each jump in 1. corresponds to an $E''$-edge traversal. Since $|E''|=\tO(m^{1/3})$, we may easily bound the total number of $E''$-edge traversals during a covering walk of $G$ by $\tO(m^{4/3})$ (use Lemma~\ref{lem:ultimate} with $V' = V$ and $\diameff(G) \le \diamdist(G) < n$). Similarly, for step 2., each jump ends with an $E'$-edge traversal, so it remains only to deal with jumping over $E'\setminus E''$-edges. There are $\tO(m^{2/3})$ of them, and each lies inside a component $F = F_i$.

Fix $F$ and let $E'_F = E' \cap E(F)$. Our aim is to bound the number of $E'_F$-traversals until $F$ is covered. This will also be a bound for the number of jumps ending with an $E'_F$-edge. To this end, we use Lemma~\ref{lem:ultimate} with $V' = F$ and the edge set being $E'_F$. Because $\diameff(V') = O(m^{2/3})$, this yields a bound of $\tO(|E'_F|\cdot m^{2/3})$.

Summing this over all $F$'s we obtain that the total time spent on $E'\setminus E''$-edge traversals caused by jumps is at most $\sum_{i=1}^{q}\tO(|E'_{F_i}| \cdot m^{2/3})=\tO(|E'\setminus E''| \cdot m^{2/3})=\tO(m^{4/3})$. This completes our argument. 
\end{proof}

The two lemmas above let us conclude that the sampling can be performed in total expected time $\tO(m^{4/3})$.
\bibliographystyle{abbrv}
\bibliography{files/references}
\appendix
\section{Appendix}

\subsection{Proof of Lemma~\ref{lem:cover_weighted}}\label{app:cover_weighted}

We prove the lower bound only; the upper bound is a special case of Lemma~\ref{lem:ultimate}.

Let us denote by $H(u,v)$ the expected number of steps before a random walk starting at $u$ visits $v$. We have $\kappa(u,v) = H(u,v) + H(v,u)$ and hence $\max(H(u,v), H(v,u)) \ge \frac 12 \kappa(u,v)$. Now one just needs to perform the simple estimation that
\[
Cov(G) \ge \max_{u,v \in V} H(u,v) \ge \frac 12 \max_{u,v \in V} \kappa(u,v) = \max_{u,v \in V} m \Reff(u,v) = m \diameff(G),
\]
where the first inequality follows as covering a graph $G$ starting from $u$ requires hitting all the other vertices first, and the first equality is by definition of effective resistance -- see Section \ref{sec:effective_resistance}.

\subsection{Proof of Lemma~\ref{lem:ball_growing}}\label{app:ball_growing}

We employ here an adaptation of the partitioning procedure from \cite{KelnerM09} that is itself based on the ball-growing technique of Leighton and Rao \cite{LeightonR99}. 

For a subset of vertices $V'\subseteq V$ let us define the set of neighbors $N(V')$ to be $V'\cup \{w\in V: (v,w) \in E \mbox{ for some }v\in V'\}$. Also recall that $\partial V'$ denotes the set of edges with exactly one endpoint in $V'$.

We start with the original $G$, $i=1$ and perform the following procedure until $G$ is empty:
\begin{itemize}\addtolength{\itemsep}{-.5\baselineskip}
	\item choose any vertex $v \in G$ and set $D:=\{v\}$;
	\item while $|\partial D|>\phi |E(D)|$ do: set $D:=N(D)$;
	\item add $D$ as a new $V_i$;
	\item remove $D$ and all incident edges from $G$;
	\item $i:=i+1$.
\end{itemize}
Obviously $V_1,\ldots, V_k$ is a partition of $V$. One can also easily prove that the condition (2) in Definition \ref{def:ball_decomp} is satisfied.  It remains to show that each piece has a small diameter, i.e., establish condition (1). 

To this end, note that the radius of $D$ at the moment it is added as a new piece is bounded by the number of the while-loop executions. We also see that each execution increases $|E(D)|$ at least by a factor of $1+\phi$. This means that after $s$ steps, $|E(D)|>(1+\phi)^s$. Since $(1+\phi)^{\log (m+1)/\phi}\geq m+1$ for $\phi<\frac 12$, we deduce that the radius of $D$ is bounded by $\frac{\log (m+1)}{\phi}$ and thus the diameter of $D$ is at most $\frac{2\log (m+1)}{\phi}$. It is not hard to see that the presented algorithm can be implemented to run in $O(m+n)$ time.

\subsection{Proof of Lemma~\ref{lem:conditioning_progress}} \label{app:conditioning_progress}

It is not hard to see that given $F'$ and $F^*$, the construction of the corresponding $(F^*,F')$-conditioned graph $G'$ can be done in $\tO(m)$ time. Also, note that as edges of $F^*$ (and thus also edges of $F'$) are always interior edges of components of the overlay $\cP(\ocC)$ of the covering family $\ocC$, removing or contracting them translates to removing or contracting the interior edges of the sets in $\ocC$. So, doing that (which can again be performed in $\tO(m)$ time, as each edge can be in at most $\ell+1=O(\log m)$ different sets) makes the resulting covering family $\ocC''$ remain valid for $G'$ and still $\alpha$-bounded. (As mentioned before, the value of $m$ used in the definition of the covering family and $\alpha$-boundedness always refers to the number of edges of the original graph and thus remains unchanged here.)

Now, observe that the only difference between the overlay $\cP(\ocC)$ of $\ocC$ and the overlay $\cP(\ocC'')$ of $\ocC''$ is that each minimum-age component $P$ of $\cP(\ocC)$ has its interior edges removed and/or contracted in $\cP(\ocC'')$. This means that each such $P$ corresponds to a minimum-age component $P'$ in $\cP(\ocC'')$ that is a collection of vertices with empty interior (the only edges incident to these vertices are the boundary edges of $P$, which now also form the boundary of $P'$). 

So, if the age $r^*$ of $\ocC$ (and thus that of $\ocC''$) was equal to $\ell$, then all the components of $\cP(\ocC)$ must have been minimum-age and thus all of them have empty interiors now. Given the way the overlay components are defined, this means that each edge remaining in $G'$ is part of the boundary of some set in $\ocC''$. Thus, in that case, we simply take $\ocC'$ equal to $\ocC''$ and the lemma follows. 

We can, therefore, restrict ourselves to the case when $r^*$ is actually strictly smaller than $\ell$. In this case, let us take $\ocC'$ to be a covering family in $G'$ that arises from adding, for each vertex of each minimum-age component $P'$ of $\cP(\ocC'')$, a singleton set containing only that vertex to the upper-most subdivision $\cC_{\ell}''$ of $\ocC''$. (Note that as $r^*<\ell$ and $P'$ is minimum-age, it does not intersect any of the sets in $\cC_{\ell}''$ and thus adding these singleton sets to $\cC_{\ell}''$ is valid.)

Observe that performing this operation results in each component $P'$ being split into single-vertex sets in the overlay $\cP(\ocC')$ of $\ocC'$ and each of these sets being of age $\ell$. So, as $r^*<\ell$, this increases the age of all the minimum-age components of $\cP(\ocC'')$ and thus the age of $\ocC'$ is at least $r^*+1$ now.

In light of the above, all that remains to be done is to argue that $\ocC'$ is still $\alpha$-bounded. To this end, note that adding our singleton sets does not introduce any new boundary edges. This is so since, as we already noted above, the only edges incident to vertices of minimum-age components $P'$ of $\cP(\ocC'')$ were boundary edges. Furthermore, as these singleton sets are added to $\cC_\ell''$, it cannot affect these edges' status as $i$-based boundary edges. That is, if such an edge was an $i$-based boundary edge in $\ocC''$ then it remains such in $\ocC'$. Thus, as the $\alpha$-boundedness property boils down to bounding the number of $i$-based boundary edges for each $0\leq i\leq \ell$, it must be preserved by our operation and thus the covering family $\ocC'$ is indeed $\alpha$-bounded, as desired.

\subsection{Proof of Lemma \ref{lem:massaging}} \label{app:massaging}

The cut $\cut$ separates $G$ into connected components (possibly more than two) -- some of them intersect $P_1$, some intersect $P_2$ and some intersect neither (but none can intersect both). Let $A_1$ be the union of those that intersect $P_1$ and $A_2$ be the union of the others; we thus obtain a decomposition $G = A_1 \cup A_2$ of $G$ where $P_1 \subseteq A_1$, $P_2 \subseteq A_2$ and all edges between $A_1$ and $A_2$ come from the cut $\cut$, i.e., $E(A_1,A_2) \subseteq \cut$.

Our plan will be to choose the set $S$ as a subset of $W(P)$ of at most half its size and thus set $i$ to be larger than $\age(P)$. (Recall that $W(P)$ denotes the age witness of $P$ -- see Section \ref{sec:marginals_overview}.) Because no set in $C^i$ with $i > \age(P)$ intersects $P$, this will make it easier for us to guarantee condition~\eqref{cond:disjoint}.

So let us define $U_1 := W(P) \cap A_1$, $U_2 := W(P) \cap A_2$ and assume wlog that $|E(U_1)| \le |E(U_2)|$. 
Since $U_1$ and $U_2$ are disjoint and $U_1 \cup U_2 = W(P)$, we get that
\begin{equation} \label{eq:fresh_color}
|E(U_1)| \le \frac 12 |E(W(P))| \le \frac{m}{2^{\age(P)+1}},
\end{equation}
where we used the bound \eqref{eq:def_size_bound_covering_family}.

The set $U_1$ is our candidate for the set $S$. We will now modify it appropriately (by shrinking) in order to satisfy all the required conditions.

Recall that we want to enforce $\frac{m}{2^{i+1}} < |E(S)| \le \frac{m}{2^i}$. Let $i$ be the right level for $U_1$, that is, the smallest $i$ such that $\frac{m}{2^{i+1}} < |E(S)|$. By~\eqref{eq:fresh_color} we know that $i > \age(P)$. We would ideally like to return $U_1$ as a candidate for insertion into $\cC_i$. However, $U_1$ may intersect some sets in $\cC_i$, which violates our condition~\eqref{cond:disjoint}. To deal with this problem, note that since $i > \age(P)$, no set in $\cC_i$ intersects $P$ (or $P_1$), so we can remove the intersection of $U_1$ with all the sets in $\cC_i$, and we still have that $U_1 \supseteq P_1$. 

Of course, now -- after removing this intersection with $\cC_i$ -- the size $|E(U_1)|$ of the interior of $U_1$ might have decreased considerably. $U_1$ still contains $P_1$, which is a large piece, so it cannot be too small, i.e., it has size at least $m^{2/3}$, but it might no longer hold that $|E(U_1)| > \frac{m}{2^{i+1}}$. If this indeed happens, then we appropriately increase $i$, and repeat.

The resulting set-refining routine works as follows. We start with $i=1$ and then:
\begin{enumerate}[(a)]\addtolength{\itemsep}{-.3\baselineskip}
	\item if $|E(U_1)| \le \frac{m}{2^{i+1}}$, increment $i$ and repeat;
	\item otherwise, if the intersection of all sets in $\cC_i$ with $U_1$ is nonempty, remove it from $U_1$ and repeat (starting from step (a));
	\item otherwise, we can stop and output the desired $S := U_1$ and $i$.
\end{enumerate}

Let us now verify the satisfaction of all conditions. Conditions~\eqref{cond:size}, \eqref{cond:disjoint} and \eqref{cond:separation} are clear given the design of our set-refining routine. To establish condition~\eqref{cond:i_0_ell}, notice that $i < \ell$ because even after all the removals, $U_1$ still contains $P_1$, hence $\frac{m}{2^{i}} \ge |E(U_1)| \ge |E(P_1)| > m^{2/3} = \frac{m}{2^{\ell}}$. Also $i > \age(P) \ge 0$.

Regarding condition~\eqref{cond:boundary_edges}, observe that all boundary edges of $U_1$ are either boundary edges of $A_1$ (and they come from the cut $\cut$), or arise as a result of intersecting $A_1$ with the age witness $W(P) \in \cC_{\age(P)}$ and taking a set difference with some number of sets from subdivisions $\cC_j$ with $j \le i$. In the latter case, they are already boundary edges of $\ocC$ on the respective levels $\age(P) < i$ and $j \le i$.

\subsection{Proof of Lemma \ref{lem:partitioning_running_time}} \label{app:partitioning_running_time}

Note that by Theorem \ref{thm:SpielmanSrivastava}, the construction of our low-dimensional embedding $\cR$ can be carried out in $\tO(m)$ time.

Let us now consider the runtime of the $Cut(P)$ procedure for one component $P \in \cP(\ocC)$. If the ball-growing does not fail, then, using Lemma \ref{lem:ball_growing}, we can bound this running time by $\tO(|P| + |E(P)|)$. On the other hand, if ball-growing fails, then this running time is equal to $\tO(m)$ plus the time needed to execute the recursive calls to $Cut(P')$ and $Cut(P'')$. 

Therefore, to bound the total running time $Cut(P)$, we model it using a binary tree whose leaves correspond to ball-growing succeeding and internal nodes describe handling ball-growing failures and subsequent recursive calls. By summing the contributions of all the nodes of this tree, we can conclude that the total time spent on all calls to the $Cut(P)$ procedure is at most
\[
\tO(m \cdot c) + \sum_{P \in \cP(\ocC)} \tO(|P| + |E(P)|) = \tO(m \cdot c + m)
\] where $c$ is the total number of cuts computed during handling ball-growing failures (the sum is taken over the final overlay $\cP(\ocC)$).

Now, every time we compute a cut, we end up adding a new set to $\cC_i$ with $0 < i < \ell$. By Lemma~\ref{lem:number_of_insertions_of_set}, this can happen at most $\sum_{i=1}^{\ell-1} 2^i = O(2^\ell) = O(m^{1/3})$ times. Hence we get a runtime of $\tO(m\cdot c + m) =\tO(m^{4/3})$.

\subsection{Proof of Lemma \ref{lem:partitioning_alpha_boundedness}} \label{app:partitioning_alpha_boundedness}

We need to prove that, for every $i$, the number of $i$-based boundary edges is at most $(\alpha  + \alpha^*) m^{1/3} 2^i$ with $\alpha^* = \log^{O(1)}m$. To this end, we will show that the number of new $i$-based boundary edges (ones that were not present in the original $\ocC$) is $\tO(m^{1/3} 2^i)$.

\textbf{Case $i < \ell$:} every new $i$-based boundary edge is the boundary edge of a new set $S \in \cC^i$ (added in the cutting case). Therefore we only need to bound the number of $i$-based boundary edges introduced by all new sets $S \in \cC_i$. If we can prove that adding a new set $S \in \cC_i$ introduces at most $O(m^{1/3})$ new $i$-based boundary edges, we will be done by Lemma~\ref{lem:number_of_insertions_of_set}. But condition~\eqref{cond:boundary_edges} of Lemma~\ref{lem:massaging} means exactly that all new $i$-based boundary edges of such a set $S$ must come from the cut $\cut$ of the phase when $S$ was added. And recall that Lemma~\ref{lem:good_cut} applied to the components $P_1$ and $P_2$, of low diameter $m^{1/3}$ and high effective resistance of $\Omega(m^{1/3})$ between them, gave us exactly the bound $m^{1/3}$ for the cardinality of the cut $\cut$, as computed in~\eqref{eq:before_good_cut_lemma}.

\textbf{Case $i = \ell$:} every new $\ell$-based boundary edge is a boundary edge of a new small piece (added when ball-growing succeeds). During the entire partitioning procedure, all the calls to the ball-growing primitive that introduce new $\ell$-based boundary edges concern disjoint subgraphs $H$ of $G$. Each such call yields a set of cut edges of size $O \left( \frac{\log m}{m^{1/3}}|E(H)| \right)$, so all the calls result in at most $\tO \left( \frac{m}{m^{1/3}} \right) = \tO(m^{2/3}) = \tO(m^{1/3} \cdot 2^\ell)$ cut edges.

\subsection{Proof of Lemma \ref{lem:cost_shortcutting_sampling}}\label{app:cost_shortcutting_sampling}

Let us first note that if the interior $E(S)$ of $S$ has more than one connected component, then the problem decomposes into completely independent subproblems corresponding to each of these connected components. So, from now on, we can assume that $E(S)$ forms a connected graph and, in particular, $|S|\leq |E(S)|+1$. 

Fix $e=(u,u')\in \partial S$ (with $u\in S$ and $u'\notin S$) and suppose we want to calculate the probabilities $P_v(e)$ for all $v\in S$. Let us construct an auxiliary graph $G'$ as follows. The vertex set is $S\cup\{u',u^*\}$, where $u^*$ is a dummy vertex. We include in $G'$ all edges from $E(S)$ and the edge $e$; we also include all boundary edges of $S$ (those going from $S$ to $V \setminus S$) except $e$, but we change their endpoint that originally belonged to $V \setminus S$  to be $u^*$ now.

Now, it is not hard to see that $P_v(e)$ is equal to the probability $q_v$ that the random walk in this auxiliary graph starting at $v$ will visit $u'$ before hitting $u^*$. The latter probability turns out, however, to be closely related to the electrical structure of the graph $G'$. 

More precisely, let us treat $G'$ as an electrical network, where every edge has unit resistance. It is known (see Chapter~4 in~\cite{Lovasz93}) that when we impose a voltage of $1$ and $0$ at $u'$ and $u^*$ respectively, the voltage at $v$ becomes exactly $q_v$ (and thus $P_v(e)$). 

Furthermore, this probability $q_v$ can be computed by solving a Laplacian linear system of the form $L'x=y$, where $L'$ is the discrete Laplacian of $G'$. In fact, in this way we can get such probabilities $q_v$ for all possible starting vertices $v\in S$ simultaneously. For details, refer to Lemma~9 in \cite{KelnerM09}.

Therefore, we see that computing all the probabilities $P_v(e)$ boils down to solving a Laplacian linear system $L'x=y$ for all the graphs $G'$ corresponding to each edge $e\in \partial S$.

As shown first by Spielman and Teng~\cite{SpielmanT04} (see also \cite{KoutisMP10,KoutisMP11,KelnerOSZ13}), solving such a system is possible in time that is nearly-linear with respect to the number of edges in $G'$. However, this method does not provide us with exact values of $P_v(e)$'s -- it outputs only approximate values with prespecified precision $\varepsilon$. So, in our case, for a given $\varepsilon>0$, the total computation takes $\tO(|E(S)||\partial S|\log 1/\varepsilon)$ time and computes the values of all $P_v(e)$s up to an additive error of $\varepsilon$.

Now, suppose for a moment that all these computed probabilities were exact. Then, for a fixed $v\in S$, it is easy to build a simple array structure which allows us to sample edges from $\partial S$ with the desired distribution. It boils down to constructing first an array of prefix sums of size $|\partial S|$ and then, whenever a sample is needed, performing  a simple binary search in this array, which takes time $O(\log m)$. (See Lemma~7 in \cite{KelnerM09} for details.)

Finally, the way of coping with the precision issue is based on an idea due to Propp \cite{Propp10}. Here, we only explain the idea on a simpler example and refer the reader to Lemma~7.3.2 in~\cite{Madry11} for a complete description. 

In this example, assume we want to take a sample from a $p$-biased Bernoulli distribution, i.e., sample a $0$-$1$ random variable that is $0$ with probability $p$ and $1$ with probability $1-p$. If we knew the value of $p$, we could perform this sampling by simply choosing a value $r\in [0,1]$ uniformly at random and then outputting $0$ if $r\leq p$ and $1$ otherwise.

Assume now that instead of the exact value of $p$ we only have access to an algorithm which, given $\varepsilon$, outputs $p$ with precision $\pm \varepsilon$ in time proportional to $T\cdot \log 1/\varepsilon$, for some $T>0$. To perform our sampling in this case, we again choose a uniformly random $r\in [0,1]$ and then compare it to the current approximation $p'\in [p-\varepsilon,p+\varepsilon]$ of the value of $p$, for some $\varepsilon>0$. 

Observe that as long as $r<p'-\varepsilon$ or $r>p'+\varepsilon$, we can simply output $0$ and $1$ respectively, as we are certain that $r<p$ or, respectively, $r>p$. Now, in the remaining case of $r\in [p'-\varepsilon, p'+\varepsilon]$, we simply keep halving the value of $\varepsilon$ (and computing the corresponding refinement of our approximation to the value of $p$) until we are able to conclude that either $r< p$ or $r>p$ (in which case we can output an appropriate value).

Clearly, this procedure terminates (with probability $1$) and outputs the correct sample. It thus remains to analyze the expected time spent on recomputing $p'$. To this end, note that the probability that during our sampling the required precision will be less than $2^{-k}$ is at most $2^{-k+1}$. So, the expected time spent on recomputing $p'$ is bounded by
\[
\sum_{k\ge 1} \frac{2}{2^k} \cdot T\cdot \log (1/2^{-k})=T \cdot \sum_{k\ge 1} \frac{2\cdot k}{2^k}=O(T).
\]
Similarly, one can easily show that if we needed $s$ samples (instead of only one) then this expected time would become $O(T \cdot \log s)$.

\subsection{Proof of Lemma~\ref{lem:ultimate}} \label{app:proof_of_ultimate}

Let us define the hitting time $H(u,v)$ between two vertices $u, v \in V$ to be the expected number of steps before a random walk starting at $u$ visits $v$. We have $\kappa(u,v) = H(u,v) + H(v,u)$.

We break the proof down into a sequence of lemmas.

\begin{lemma}[cf. Lemma~2.8 of~\cite{Lovasz93}] \label{lem:LovaszCoverLog}
Let $V' \subseteq V$ be a subset of vertices, $s \in V$ -- a~starting vertex, and let $M = \max_{v \in V'} H(s,v)$ be the maximum hitting time between $s$ and a~vertex in $V'$. Suppose we run a random walk from $s$ until more than half of the vertices from $V'$ are visited. The expected length of this walk is at most $2 M$.
\end{lemma}
\begin{proof}
We assume first that $|V'| = 2k+1$ is odd. Let $\alpha_v$ be the time when vertex $v$ is first visited. Then the length of the walk, $\beta$, is the $(k+1)$-th largest of the $\alpha_v$. Hence
\[(k+1) \beta \le \sum_{v \in V'} \alpha_v\]
and
\[\bE[\beta] \le \frac{1}{k+1} \sum_{v \in V'} \bE[\alpha_v] \le \frac{1}{k+1}\cdot  |V'| \cdot M = \frac{2k+1}{k+1} M < 2 M.
\]
If $|V'| = 2k$ is even, we proceed analogously and obtain $\frac{2k}{k}M\leq 2M$ at the end.
\end{proof}

\begin{lemma} \label{lem:preUltimate}
Let $V' \subseteq V$. Suppose we run a random walk in $G$ from a vertex $v \in V'$ and stop it on the first visit to $v$ after all vertices from $V'$ have been visited. Then the expected number of steps is \[O(|E| \cdot \diameff(V') \cdot \log n).\]
\end{lemma}
\begin{proof}
The proof is similar to that of Theorem~2.7 of~\cite{Lovasz93}. We divide the walk into phases. The first one ends when more than half of vertices from $V'$ have been visited. By Lemma~\ref{lem:LovaszCoverLog} (with $s = v$), we know that its expected length is at most \[
2 \max_{u \in V'} H(v,u) \le 2 \max_{u \in V'} \kappa(v,u) = 4 |E| \max_{u \in V'} \Reff(v,u) = O(|E| \diameff(V'))
\] 
(recall that we defined $\Reff(v,u)$ to be $\frac{\kappa(v,u)}{2|E|}$).

The second phase ends when at least half of the remaining vertices of $V'$ have been covered, and the lemma applied to that set of vertices (and the starting vertex $s$ being the vertex where the first phase ended) again bounds the expected length of the second phase by the same quantity. (It is important here that $s \in V'$, as the first phase must have ended with a visit to a vertex of $V'$.) We proceed in this fashion until $V'$ is covered. By linearity of expectation, the expected length of the walk up to this point is $O(|E| \diameff(V'))$ times the number of phases, which is clearly $\log |V'| = O(\log n)$. (Note that the number of phases is a fixed quantity here.) Finally, the remaining (expected) time needed to return to $v$ from the last visited vertex of $V'$ is at most $\max_{u \in V'} H(u,v) = O(|E| \diameff(V'))$.
\end{proof}

Now, we will argue that every edge in $G$ has exactly the same contribution to the expected length of the random walk of Lemma~\ref{lem:preUltimate}. Intuitively, this is because the (infinite) stationary random walk on $G$ (where, in the long run, every edge has the same contribution) can be viewed as an (infinite) sequence of samples from the trajectory of our random walk (plus an initial fragment of bounded expected length). Thus, if any edge contributed more than a $1/|E|$-fraction to the length of our walk, then this discrepancy would also emerge in the stationary walk, which is a contradiction.

We proceed to formalizing the above intuition. To this end, let us state the following lemma from renewal theory. It appears as Proposition~7.3 in~\cite{Ross09}.
\begin{lemma} \label{lem:renewal_theory}
Let $X_1, X_2, \ldots$ be an infinite sequence of nonnegative i.i.d. random variables with finite expectation, i.e., $X_1 \ge 0$, $\bE[X_1] < \infty$. Let $R_1, R_2, \ldots$ be another such sequence. (Each sequence is independent, but $R_i$ may depend on $X_i$ for every $i$.) For any $t \ge 0$, set $R(t) = R_1 + R_2 + \dots + R_{N(t)}$, where $N(t)$ is the maximum integer such that $X_1 + X_2 + \dots + X_{N(t)} \le t$. Then
\[ \lim_{t \to \infty} \frac{\bE[R(t)]}{t} = \frac{\bE[R_1]}{\bE[X_1]}. \]
\end{lemma}

We are now ready to analyze the number of traversals over a given edge $e$ in our random walk. 

\begin{lemma} \label{lem:ultimate_for_one_edge}
Let $V' \subseteq V$ and $v \in V'$. Also, fix any edge $e \in E$. Consider again the random walk from the statement of Lemma~\ref{lem:preUltimate}. The expected number of times  $e$ will be traversed during this walk is $O(\diameff(V') \cdot \log n)$.
\end{lemma}
\begin{proof}
Consider an infinite random walk starting at $v$. We will divide this walk into phases, each of which starts at $v$ and ends on the first visit to $v$ after all vertices from $V'$ have been visited. In this way we get an i.i.d. sequence of copies of the random walk of Lemma~\ref{lem:preUltimate}. For every $i$ let $X_i$ be the length of the $i$-th phase, and let $R_i$ be the number of times the edge $e$ was traversed during this phase. By Lemma~\ref{lem:renewal_theory} we obtain that
\[
\frac{\mbox{the expected number of traversals of edge $e$ during one phase}}{\mbox{the expected length of one phase}}
=
\frac{\bE[R_1]}{\bE[X_1]} = \lim_{t \to \infty} \frac{\bE[R(t)]}{t}, \]
where $R(t)$ is the number of traversals of $e$ during all phases which have finished by time $t$. We need to show that this quantity is equal to $1/|E|$.

Indeed, if we define $R'(t)$ to be the number of traversals of $e$ up to time $t$, then it is well-known
that $\lim\limits_{t \to \infty} \frac{\bE[R'(t)]}{t} = \frac{1}{|E|}$ (which is the stationary probability of traversing $e$). Furthermore, the difference between $R'(t)$ and $R(t)$ is just the number of traversals of $e$ during the current phase at time $t$, and we can upper-bound $\bE[R'(t) - R(t)]$ by $\bE[X_1]$. Hence, we have that
\[
\frac{\bE[R'(t)]}{t} \ge \frac{\bE[R(t)]}{t} = \frac{\bE[R'(t)] - \bE[R'(t) - R(t)]}{t} \ge \frac{\bE[R'(t)] - \bE[X_1]}{t}.
\]
Taking a limit of $t\rightarrow \infty$ of both sides concludes the proof.
\end{proof}

Now Lemma~\ref{lem:ultimate} follows easily. Namely, let us consider a random walk started at an arbitrary vertex (possibly outside of $V'$) and let $v_0$ be the first vertex in $V'$ on that walk. (No edges from $E'$ will be traversed before $v_0$ is reached, as $E' \subseteq E(V')$.) Applying Lemma~\ref{lem:ultimate_for_one_edge} with $v=v_0$ to each edge $e \in E'$, we obtain that the total number of traversals of edges in $E'$ in this random walk until the set $V'$ is covered is $O(|E'|\cdot \diameff(V') \cdot \log n)$, as we wanted to show.

\end{document}